\documentclass[11pt,letter]{amsart}
\usepackage{rotate}

\usepackage{mathdots}
\usepackage{color}
\usepackage{enumerate,graphicx,stfloats,graphics,xspace}
\usepackage{subfigure} 
\usepackage{tikz}
\usetikzlibrary{calc}
\usetikzlibrary{decorations.pathreplacing}
\usetikzlibrary{shapes,snakes}
\newcommand{\NPTime}{\mbox{\sc{NPTime}}}%

\newcommand{\fA}{\mathfrak{A}}
\newcommand{\fB}{\mathfrak{B}}

\newcommand{\fc}{\mathfrak{c}}

\newcommand{\fd}{\mathfrak{d}}

\renewcommand{\phi}{\varphi} % Nicer-looking phi
\renewcommand{\rho}{\varrho} % Nicer-looking phi

\newcommand{\NExpTime}{\textsc{NExpTime}}

\newcommand{\N}{{\mathbb N}}   % Natural numbers
\newcommand{\kint}[1]{\ensuremath{\mbox{int}_{#1}}}
\newcommand{\kin}[1]{\ensuremath{\mbox{in}_{#1}}}
\newcommand{\kout}[1]{\ensuremath{\mbox{out}_{#1}}}
\newcommand{\kinStar}[1]{\ensuremath{{\mbox{in}^\triangleleft_{#1}}}}
\newcommand{\koutStar}[1]{\ensuremath{{\mbox{out}^\triangleleft_{#1}}}}
\newcommand{\keq}[1]{\ensuremath{\mbox{eq}_{#1}}}
\newcommand{\kdeq}[2]{\ensuremath{\mbox{eq}_{#1, #2}}}
\newcommand{\keqDig}[1]{\ensuremath{\mbox{eqDig}_{#1}}}
\newcommand{\kdeqDig}[2]{\ensuremath{\mbox{eqDig}_{#1, #2}}}
\newcommand{\kpred}[1]{\ensuremath{\mbox{pred}_{#1}}}
\newcommand{\kdpred}[2]{\ensuremath{\mbox{pred}_{#1,#2}}}
\newcommand{\kpredDig}[1]{\ensuremath{\mbox{predDig}_{#1}}}
\newcommand{\kdpredDig}[2]{\ensuremath{\mbox{predDig}_{#1,#2}}}
\newcommand{\vertex}{\ensuremath{\mbox{vtx}}}

\newcommand{\tower}[2]{\ensuremath{\mathfrak{t}(#1,#2)}}
\newcommand{\val}[1]{\ensuremath{\mbox{val}_{#1}}}
\newcommand{\FLk}[1]{\ensuremath{\mathcal{FL}}^{#1}}
\newcommand{\FL}{\ensuremath{\mathcal{FL}}}

\newcommand{\kzero}[1]{\ensuremath{\mbox{zero}_{#1}}}

\newcommand{\ianIgnore}[1]{}

\newtheorem{theorem}{Theorem}
\newtheorem{lemma}[theorem]{Lemma}

\newcommand{\sizeOf}[1]{\lVert #1 \rVert}
\newcommand{\eqk}[2]{\ensuremath{\equiv}^{}_{}}

\newcommand{\tpA}{\ensuremath{{\rm tp^{\fA}}}}
\newcommand{\tpB}{\ensuremath{{\rm tp^{\fB}}}}
\newcommand{\tpBp}{\ensuremath{{\rm tp^{\fB'}}}}

\newcommand{\ConA}{\ensuremath{{\rm Con^{\fA}}}}
\newcommand{\ftp}{\ensuremath{{\rm ftp}}}

\newcommand{\myWedge}{\ensuremath{\hspace*{-1mm}\wedge\hspace*{-1mm}}}
\newcommand{\myRightarrow}{\ensuremath{\hspace*{-1mm}\rightarrow\hspace*{-1mm}}}

\title{Quine's Fluted Fragment Revisited}

\author{I. Pratt-Hartmann}

\address{School of Computer Science\\ Manchester University\\
Manchester M13 9PL, UK\and Instytut  Informatyki\\
Uniwesytet Opolski\\45--040 Opole, Poland}
\email{ipratt@cs.man.ac.uk}
\author{Wies\l{}aw Szwast}
\address{Instytut  Informatyki\\
Uniwesytet Opolski\\45--040 Opole, Poland}
\email{szwast@math.uni.opole.pl}

\author{Lidia Tendera}
\address{Instytut Informatyki\\
Uniwesytet Opolski\\45--040 Opole, Poland}
\email{tendera@math.uni.opole.pl}

\begin{document}

\maketitle
\begin{abstract}
We study the fluted fragment, a decidable fragment of first-order logic with an unbounded number of variables,
originally identified in 1968 by W.V.~Quine. We show that the satisfiability problem for this fragment has non-elementary complexity, thus refuting an earlier published claim by W.C. Purdy that it is in \NExpTime.
More precisely, we consider $\FL^m$, the intersection of the fluted fragment and the $m$-variable fragment of first-order logic, for all $m \geq 1$. We show that, for $m \geq 2$, this sub-fragment forces $\lfloor m/2\rfloor$-tuply exponentially large models, and that its satisfiability problem is $\lfloor m/2\rfloor$-\NExpTime-hard. We further establish that, for $m \geq 3$, any satisfiable $\FL^m$-formula has a model of at most ($m-2$)-tuply exponential size, whence the satisfiability (= finite satisfiability) problem for this fragment is in ($m-2$)-\NExpTime. Together with other, known, complexity results, this provides tight complexity bounds for $\FL^m$ for all $m \leq 4$. 
\end{abstract}

\section{Introduction}
The fluted fragment, here denoted $\FL$, is a fragment of first-order logic in which, roughly speaking, the order of quantification of variables coincides with the order in which those variables appear as arguments of predicates. Fluted formulas arise naturally as first-order translations of quantified English sentences
in which no quantifier-rescoping occurs, thus:
\begin{align}
& \mbox{
\begin{minipage}{10cm}
\begin{tabbing}
No student admires every professor\\
$\forall x_1 (\mbox{student}(x_1) \rightarrow \neg \forall x_2 (\mbox{prof}(x_2) \rightarrow \mbox{admires}(x_1, x_2)))$
\end{tabbing}
\end{minipage}
}
\label{eq:eg1}\\
& \mbox{
\begin{minipage}{10cm}
\begin{tabbing}
No lecturer introduces any professor to every student\\
$\forall x_1 ($\=$\mbox{lecturer}(x_1) \rightarrow$
 $\neg \exists x_2 ($\=$\mbox{prof}(x_2)
\wedge$\\
\> $\forall x_3 (\mbox{student}(x_3) \rightarrow \mbox{intro}(x_1,x_2,x_3))))$.
\end{tabbing}
\end{minipage}
}
\label{eq:eg2}
\end{align}

The origins of the fluted fragment can be traced to a paper given by W.V.~Quine to the 1968 {\em International Congress of Philosophy}~\cite{purdy:quine69}, in which the author defined what he called the {\em homogeneous $m$-adic formulas}.
In these formulas, all predicates have the same arity $m$, and all atomic formulas have the same argument sequence $x_1, \dots, x_m$. Boolean operators and quantifiers may be freely applied, except that the order of quantification must follow the order of arguments: a quantifier binding an occurrence of $x_i$ may only be applied to a subformula in which all occurrences of $x_{i+1}, \dots, x_m$ are already bound. Quine explained how Herbrand's decision procedure for monadic first-order logic extends to cover all homogeneous $m$-adic formulas.

\noindent
The term {\em fluted logic} first appears (to the present authors' knowledge) in \linebreak
Quine~\cite{purdy:quine76b}, where the restriction that all predicates have the same arity is abandoned, a relaxation which, according to Quine,
does not affect the proof of decidability of satisfiabilty. The allusion is presumably architectural: we are invited to think of arguments of predicates as being `lined up' in columns.
Quine's motivation for defining the fluted fragment was to locate the boundary of decidability in the context of his reconstruction of first-order logic in terms of \textit{predicate-functors},
which Quine himself described as a `modification of Bernays' modification of Tarski's cylindrical algebra'~\cite[p.~299]{purdy:quine76a}. Specifically, the fluted fragment can be identified by dropping from full predicate functor logic those functors associated with the permutation and identification of variables, while retaining those concerned with cylindrification and Boolean combination.

Notwithstanding its predicate-functorial lineage, the fluted fragment has, as we shall see, a completely natural characterization within the standard r\'{e}gime of bound variable quantification, and thus constitutes an interesting fragment of first-order logic in its own right. In fact, $\FL$ overlaps in expressive power with various other such fragments. For example, \textit{Boolean modal
	logic} (Lutz and Sattler~\cite{purdy:ls01})
maps, under the standard first-order translation, to $\FL$---in fact, to $\FLk{2}$, the fluted fragment restricted to just two variables.
On the other hand, even $\FLk{2}$ is not contained within the so-called {\em guarded fragment} of first-order logic (Andr\'{e}ka, van Benthem and N\'{e}meti~\cite{purdy:avbn98}): the formula~\eqref{eq:eg1}, for example, is not equivalent to any guarded formula. A more detailed comparison of the fluted fragment to other familiar decidable fragments can be found in Hustadt, Schmidt and Georgieva~\cite{purdy:hsg04}.

Quine never published a proof of his later claims regarding the full fluted fragment; indeed, Noah \cite{purdy:noah80} later claimed that, on the contrary, Herbrand's technique does not obviously extend from homogeneous $m$-adic logic to the fluted fragment, and that consequently, the decidability of the satisfiability problem for the latter should be regarded as open. This problem---together with the corresponding problems for various extensions of the fluted fragment---was considered in a series of papers in the 1990s by W.C.~Purdy~\cite{purdy:purdy96b,purdy:purdy96a,purdy:purdy99,purdy:purdy02}. The decidability of $\FL$ is proved in~\cite{purdy:purdy96a}, while
in~\cite[Corollary~10]{purdy:purdy02} it is claimed that this fragment has the exponential-sized model property: if a fluted formula $\phi$ is satisfiable, then it is satisfiable over a domain of size bounded by an exponential function of the number of symbols in $\phi$. Purdy concluded~\cite[Theorem~13]{purdy:purdy02} that the satisfiability problem for $\FL$ is \NExpTime-complete.

These latter claims are false. In the sequel, we show that, for $m \geq 2$, the fluted fragment restricted to just $m$ variables, denoted $\FLk{m}$, can force models of $(\lfloor m/2 \rfloor)$-tuply exponential size, and that its
satisfiability problem is $(\lfloor m/2 \rfloor)$-\NExpTime-hard. It follows that there is no elementary bound on the size of models of satisfiable fluted formulas, and that the satisfiability problem for $\FL$ is non-elementary.\footnote{An extended abstract with this result was published in \cite{P-HST16}.} On the other hand, we also show that, for $m \geq 3$, any satisfiable formula of the $m$-variable fluted fragment has a model of $(m-2)$-tuply exponential size, so that the satisfiability problem for this sub-fragment is contained in $(m-2)$-\NExpTime. Thus, $\FL$ has the finite model property, and  its satisfiability (= finite satisfiability) problem is decidable, but not elementary. In the case $m =2$, $\FL^2$ is contained within the 2-variable fragment of first-order logic, whence its satisfiability problem is in $\NExpTime$ by the well-known result of Gr\"{a}del, Kolaitis and Vardi~\cite{purdy:gkv97}, which matches the lower bound reported above. The 
fragment $\FL^0$ is evidently the same as propositional logic, and the 
fragment $\FL^1$ likewise coincides with the 1-variable fragment of first-order logic, so that both these fragments have
\NPTime-complete satisfiability problems. 
Counting ``$0$-tuply exponential'' as a synonym for ``polynomial'',  we see that for $0 \leq m \leq 4$, $\FL^m$ is 
$(\lfloor m/2 \rfloor)$-\NExpTime-complete. For $m >4$, the above complexity bounds for 
$\FLk{m}$  leave a gap between $(\lfloor m/2 \rfloor)$-\NExpTime{} and $(m-2)$-\NExpTime.

We mention at this point another incorrect claim by Purdy concerning an extension of the fluted fragment. In Purdy~\cite{purdy:purdy99}, the author  considers what he calls {\em extended fluted logic} (EFL), in which, in addition to the usual predicate functors of fluted logic, we have an \textit{identity functor} (essentially: the equality predicate), {\em binary conversion} (the ability to exchange arguments in binary atomic formulas) and {\em functions} (the requirement that certain specified binary predicates be interpreted as the graph of a function.) Purdy claims (Corollary~19, p.~1460) that EFL has the \textit{finite model property}: if a formula of this fragment is satisfiable, then it is satisfiable over a finite domain. But EFL evidently contains the formula
\begin{equation*}
\forall x_1 \forall x_2 (r(x_1,x_2)  \rightarrow f(x_1,x_2)) \wedge
 \exists x_1 \forall x_2 \neg r(x_1, x_2) \wedge  \forall x_1 \exists x_2 r(x_2, x_1),
\end{equation*}
where $f$ is required to be interpreted as the graph of a binary function;
and this is an axiom of infinity. In view of these observations, it seems only prudent to treat Purdy's series of articles with caution.

An independent, resolution-based decision procedure for the fluted fragment 
was presented by Schmidt and Hustadt~\cite{purdy:SH00}.
No complexity bounds are given. Moreover, that paper omits detailed proofs, and these have, to the present authors' knowledge, never been published. 

In the sequel, we show that,
for $m \geq 3$, 
the satisfiability problem for $\FL^m$ is in $(m-2)$-\NExpTime. 
Specifically, we
use a model-construction-based technique to show that any satisfiable formula of $\FL^3$ has a model
of size bounded by an exponential function of the size of $\phi$; 
and we use resolution theorem-proving to reduce the satisfiability problem for $\FL^m$ ($m \geq 1$)
to the corresponding problem for
$\FL^{m-1}$, at the cost of an exponential increase in the signature and the size of the formula.
Our proof is shorter and more perspicuous than the arguments for the decidability of the satisfiability problem for $\FL$ given  in either Purdy~\cite{purdy:purdy96a} or Schmidt and Hustadt~\cite{purdy:SH00}, and yields better complexity bounds than could---in the absence of non-trivial refinements---be derived from those approaches.

The structure of this paper is as follows. Section~\ref{sec:prelim} gives some basic definitions. In Section~\ref{sec:lower}, we show that formulas of $\FLk{2m}$ can force models of $m$-tuply exponential size, and indeed that the satisfiability problem for $\FLk{2m}$ is $m$-\NExpTime-hard, thus disproving the
results claimed in Purdy~\cite{purdy:purdy02}. In Section~\ref{sec:upper}, 
we show that, for $m \geq 3$,
any satisfiable formula of $\FL^m$ has a model of size at most $(m-2)$-tuply exponential in the size of $\phi$, and hence that the satisfiability (= finite satisfiability) problem for $\FL^m$ is in $(m-2)$-\NExpTime.

\section{Preliminaries}
\label{sec:prelim}
Fix a sequence of variables $\bar{x}_\omega= x_1, x_2, \ldots$ \ . Let $\sigma$ be a purely relational signature $\sigma$---i.e.,~a signature containing predicates of any arity (including 0), but
no function-symbols or individual constants. A \textit{fluted atomic formula} (or: \textit{fluted atom}) of $\FL^{[k]}_\sigma$ is an expression $p(x_\ell, \ldots, x_k)$, where $\ell \leq k+1$, $p \in \sigma$ has arity $(k-\ell +1)$, and $x_\ell, \dots, x_k$ is a contiguous subsequence of $\bar{x}_\omega$. If $\ell = k+1$, then $p$ has arity 0---in other words, is
a propositional letter. A {\em fluted literal} of $\FL^{[k]}_\sigma$ is either a fluted atom of $\FL^{[k]}_\sigma$ or the negation of such. 
We define the sets of formulas $\FL^{[k]}_\sigma$ (for $k \geq 0$) over $\sigma$ by structural induction as follows:
(i) any fluted atom of $\FL^{[k]}_\sigma$ is a formula of $\FL^{[k]}_\sigma$; 
(ii) $\FL^{[k]}_\sigma$ is closed under boolean combinations;
(iii) if $\phi$ is in $\FL^{[k+1]}_\sigma$, then $\exists x_{k+1} \phi$ and $\forall x_{k+1} \phi$
are in $\FL^{[k]}_\sigma$. We normally suppress reference to $\sigma$, writing $\FL^{[k]}$ for $\FL^{[k]}_\sigma$.
In this context, a fluted atom of $\FL^{[k]}_\sigma$ will simply be called a \textit{fluted $k$-atom}, and similarly for
literals.
In the definition of fluted $k$-atoms, the case $\ell = k+1$ implies that if $p \in \sigma$ is a proposition letter, 
then $p$ is a fluted $k$-atoms for all $k\geq 0$ (and hence is a fluted $k$-literal and indeed a formula of $\FL^{[k]}$. 
The set of \textit{fluted formulas} is defined as \smash{$\FL = \bigcup_{k\geq 0} \FL^{[k]}$}. A \textit{fluted sentence} is a fluted formula over an empty set of variables, i.e.~an element of $\FL^{[0]}$.
Thus, when forming Boolean combinations in the fluted fragment, all the combined formulas must have as
their free variables some contiguous sub-word of $\bar{x}_\omega$; and when quantifying, only the free variable with highest index
may be bound. Note however that proposition letters may occur freely in fluted formulas.

Denote by $\FLk{m}$ the sub-fragment of $\FL$ consisting of
those formulas featuring at most $m
$ variables, free or bound.
Do not confuse $\FLk{m}$ (the set of fluted formulas with $m$ variables, free or bound) with $\FL^{[m]}$ (the set of fluted formulas with $m$ {\em free} variables). These are of course, quite different.

To avoid tedious repetition of formulas, we write $\pm \phi$ to stand ambiguously for the
formulas $\phi$ and $\neg \phi$. However, we adopt the convention that multiple occurrences of the symbol $\pm$ in a displayed formula are all resolved in the same way. Thus, for example, the 
expression 
\begin{equation*}
\bigwedge_{i=0}^{n-1} \ \bigwedge_{\ell=0}^{L} \forall x_1 (\kint{1}(x_1) \wedge \pm p_i(x_1) \rightarrow
\forall x_2  \cdots \forall x_{\ell+1} \pm p_i^{\ell}(x_1, \dots, x_{\ell+1})),
\end{equation*}
%which occurs as~\eqref{eq:phi2} in the sequel, 
stands for a \textit{pair} of $\FL^{L+1}$-formulas: one formula in which all $2n(L+1)$ occurrences of the symbol $\pm$ are deleted, and another in which they are all replaced by the symbol $\neg$.

We make extensive use of the tetration function
$\tower{k}{n}$, defined, for $n, k \geq 0$, by induction as follows:
\begin{align*}
\tower{0}{n} & = {n}\\
\tower{k+1}{n} &= 2^{\tower{k}{n}}.
\end{align*}
Thus, $\tower{1}{n} = 2^n$, $\tower{2}{n} = 2^{2^n}$, and so on. 

\section{Lower bound}

\label{sec:lower}
In this section, we establish lower complexity bounds for the fluted fragment.
Theorem~\ref{thm:bigModels} shows that
an $\FL^{2m}$-formula of size $O(n^2)$
can force models of size at least $\tower{m}{n}$, thus contradicting Corollary~10 of Purdy~\cite{purdy:purdy02}.
Theorem~\ref{thm:lower} shows that the satisfiability problem for
\smash{$\FL^{2m}$} is $m$-\NExpTime-hard, thus contradicting Theorem~11 of Purdy~\cite{purdy:purdy02}.

As a preliminary, for any $z \geq 0$, we take the (\textit{canonical}) {\em representation} of any integer $n$ in the range ($0 \leq n < 2^{z}$) to be the bit-string
$\bar{s} = s_{z-1}, \dots, s_0$ of length $z$, where $n = \sum_{i=0}^{z-1} s_{i} \cdot 2^i$. (Thus, $s_0$ is the least significant bit.)
Where $z$ is clear from context, this representation is unique.
Observe that, if, in addition, an integer $n'$ in the same range is represented by $s'_{z-1}, \dots, s'_0$,
then $n' = n-1 \mod 2^z$ if and only if, for all $i$ ($0 \leq i < z$):
\begin{equation*}
s'_i =
\begin{cases}
1- s_i & \text{if, for all $j$ ($0 \leq j < i$), $s_j = 0$;}\\
s_i    & \text{otherwise}.
\end{cases}
\end{equation*}
This simple observation---effectively, the algorithm for decrementing an integer represented in binary---will feature at various points in the proof of the following theorem. 

\begin{theorem}
\label{thm:bigModels}
For all $m \geq 1$, there exists a sequence of satisfiable
sentences $\{\phi_n\}_{n \in \N} \in \FL^{2m}$ 
such that $\sizeOf{\phi_n}$ grows polynomially with $m$ and $n$ \textup{(}and indeed quadratically in $n$
for fixed $m$\textup{)}, but the
smallest satisfying model of $\phi_n$ has at least $\tower{m}{n}$ elements. Hence, there is no elementary
bound on the size of models of satisfiable sentences in $\FL$.
\end{theorem}
\begin{proof}
Fix positive integers $m$ and $n$.
Consider a signature $\sigma_{m,n}$ featuring:
\begin{itemize}
	\item[-] unary predicates $p_0, \dots, p_{n-1}$;
	\item[-] for all $k$ in the range $1 \leq k \leq m$, a unary predicate $\kint{k}$;
	\item[-] for all $k$ in the range $1 \leq k < m$, binary predicates $\kin{k}$, $\kout{k}$.
\end{itemize}
(We shall add further predicates to $\sigma_{m,n}$ in the course of the proof.)
When working within a particular structure, we call any element satisfying the unary predicate $\kint{k}$ in that structure a $k$-\textit{integer}. Each $k$-integer, $b$, will be associated with an integer {value}, $\val{k}(b)$, between 0 and  $\tower{k}{n}-1$.
For $k=1$, this value will be encoded by $b$'s satisfaction of the unary predicates $p_0, \dots, p_{n-1}$. Specifically, for any 1-integer $b$, define
 $\val{1}(b)$ to be the integer canonically represented by the $n$-element
bit-string $s_{n-1}, \dots, s_0$, where, for all $i$ ($0 \leq i < n$),
\begin{equation*}
s_i =  \begin{cases}
1 & \text{if $\fA \models p_i[b]$};\\
0 & \text{otherwise.}
\end{cases}
\end{equation*}
On the other hand, if $b$ is a ($k+1$)-integer ($k \geq 1$), then $\val{k+1}(b)$ will be encoded by
the way in which the various $k$-integers are related to $b$ via the
predicate $\kin{k}$. Specifically,
for any $k$ ($1 \leq k <m$) and any ($k+1$)-integer $b$, define $\val{k+1}(b)$ to be
the integer canonically represented by the bit-string $s_{N-1}, \dots, s_0$ of length $N = \tower{k}{n}$
where, for all $i$ ($0 \leq i < N$),
\begin{equation*}
s_i =
\begin{cases}
1 & \text{if $\fA \models \kin{k}[a,b]$ for some ($k$)-integer $a$
 s.t.~$\val{k}(a) = i$;}\\
0 & \text{otherwise}.
\end{cases}
\end{equation*}
We shall be interested in the case where $\fA$ satisfies the following property, for all $k$ ($1 \leq k \leq m$).
\begin{description}
\item[$k$-{covering}] The function $\val{k}: \kint{k}^\fA \rightarrow [0,\tower{k}{n}-1]$ is surjective.
\end{description}
For technical reasons, we create a duplicate (mirror image) encoding of $\val{k+1}(b)$
in terms of the way in which $b$ is related to the various $k$-integers via the
predicate $\kout{k}$. Specifically, we shall be interested in the case where $\fA$ satisfies the following property, for all $k$ ($1 \leq k \leq m$)
\begin{description}
\item[$k$-{harmony}] If $k > 1$, then, for all $k$-integers $b$ and all ($k-1$)-integers $a$, $a'$
 in $\fA$ such that $\val{k-1}(a)= \val{k-1}(a')$, 
\mbox{$\fA \models \kin{k-1}[a,b] \Leftrightarrow \fA \models \kout{k-1}[b,a']$}.
\end{description}

Let us pause to consider the properties of $k$-{covering} and $k$-{harmony} for various values of $k$.
If $k < m$, $k$-covering ensures that, when we want to know what the $i$th bit in the canonical binary representation of a ($k+1$)-integer $b$ is (where $0 \leq i < \tower{k}{n}$), then  {\em there exists} a $k$-integer $a$
such that $\val{k}(a) = i$, and for which we can ask whether $\fA \models \kin{k}[a,b]$.
Conversely,
($k+1$)-harmony ensures that, if there are \textit{many} $k$-integers $a$ satisfying $\val{k}(a) = i$, then it does not matter which one we consult. For
if $\val{k}(a)= \val{k}(a')$, then by two applications of  ($k+1$)-harmony,
$\fA \models \kin{k}[a,b] \Leftrightarrow \fA \models \kout{k}[b,a] \Leftrightarrow \fA \models \kin{k}[a',b]$.

Our strategy will be to construct a satisfiable $\FLk{2m}$-formula $\Phi_{m,n}$
in the signature $\sigma_{m,n}$ such that any model $\fA \models \Phi_{m,n}$
satisfies $k$-{covering} and $k$-{harmony} for all $k$ ($1 \leq k \leq m$). It then follows from $m$-covering that
$|A| \geq \tower{m}{n}$, proving the theorem.
The signature $\sigma_{m,n}$ will feature several auxiliary predicates. In particular, we take $\sigma_{m,n}$ to contain:
\begin{enumerate}[($i$)]
\item
the unary predicates $\kzero{1}, \dots, \kzero{m}$; 
\item the binary predicates $\kdpred{1}{0}, \dots, \kdpred{m}{0}$;
\item
the ternary predicates $\kdpred{1}{1}, \dots, \kdpred{m-1}{1}$;
\item
for all $k$ ($1 \leq k \leq m$) and all $\ell$ ($0 \leq \ell \leq 2(m-k)$), the ($\ell+2$)-ary predicate $\kdeq{k}{\ell}$.
\end{enumerate}
Further predicates in $\sigma_{m,n}$ will be introduced later, as and when they are needed.
Observe from  ($iii$) that, setting $m=1$, the list of ternary predicates
in $\sigma_{1,n}$ is empty. This is as it should be: 
if $m=1$, the formula $\Phi_{m,n}$ we want to construct must lie in $\FL^2$,
and thus may not use any ternary predicates. Observe also in this regard that, in ($iv$), as $k$ increases from $1$ to $m$, the maximal value of the index $\ell$ in the predicates $\kdeq{k}{\ell}$ decreases, in steps of 2, from $2m-2$ down to $0$; hence the
maximal arity of these predicates decreases from $2m$ to $2$, whence these predicates may all be used in $\FL^{2m}$-formulas.

We show that any model $\fA \models \Phi_{m,n}$ satisfies---in addition to
$k$-covering and $k$-harmony---the following properties for all $k$ ($1 \leq k \leq m$)
concerning the interpretation of these predicates.
\begin{description}
\item[$k$-{zero}] For all $k$-integers $b$,
$\fA \models \kzero{k}[b] \Leftrightarrow \val{k}(b) = 0$.
\item[$k$-{equality}]
For all $\ell$  ($0 \leq \ell \leq 2(m-k)$), all $k$-integers $b$, $b'$ and
all $\ell$-tuples of elements $\bar{c}$,
$\fA \models \kdeq{k}{\ell}[b, \bar{c}, b'] \Leftrightarrow
\val{k}(b) = \val{k}(b')$.
\item[$k$-{predecessor}] For all $\ell$  ($0 \leq \ell \leq \min(m-k,1)$), all $k$-integers $b$, $b'$
and
all $\ell$-tuples of elements $\bar{c}$, $\fA \models \kdpred{k}{\ell}[b,\bar{c},b'] \Leftrightarrow \val{k}(b') = \val{k}(b)-1$, modulo $\tower{k}{n}$.
\end{description}
Note that, in the definition of  the property $k$-predecessor, $\bar{c}$ is either the empty sequence or a singleton.
Indeed, the bounds on $\ell$ amount to saying that $\ell$ takes values
0 or 1, except when $k =m$, in which case it takes only the value $0$. (Recall that $\sigma_{m,n}$ does not feature the predicate $\kdpred{m}{1}$.) 

Thus, in a structure satisfying $k$-{zero},
$\kzero{k}(x_1)$  can be read as ``$x_1$ is zero''; and
in a structure satisfying $k$-{predecessor}, $\kdpred{k}{0}(x_1, x_2)$, as ``$x_2$ is the predecessor of $x_1$'', and
$\kdpred{k}{1}(x_1, x_2, x_3)$ as ``$x_3$ is the predecessor of $x_1$''. Notice that, in the latter
case, the argument $x_2$ is semantically inert. Similarly, in a structure satisfying $k$-{equality},
$\kdeq{k}{\ell}(x_1, \dots, x_{\ell+2})$
can be read as ``$x_1$ is equal to $x_{\ell+2}$'', with the $\ell$ arguments $x_2, \dots, x_{\ell+1}$ again semantically inert. When naming predicates, we employ the convention that the first subscript, $k$, serves as a reminder that its primary arguments are typically assumed to be $k$-integers; the second subscript, $\ell$, indicates that $\ell$ (possibly 0) semantically inert arguments have been inserted between the primary arguments.

To prove that any model $\fA \models \Phi_{m,n}$ satisfies the properties of 
$k$-covering, $k$-harmony, $k$-zero, $k$-equality and $k$-predecessor
for all $k$ ($1 \leq k \leq m$), we proceed by induction on $k$. For ease of reading, we introduce the
various conjuncts of $\Phi_{m,n}$ as they are required in the proof. Appeals to the inductive hypothesis are indicated by the initials IH.

\medskip
\noindent
{\em Base case $(k=1)$\textup{:} }
Let $b$ be a 1-integer, and
recall that $\val{1}(b)$ is defined by $b$'s satisfaction of the predicates
$p_0, \dots, p_{n-1}$.
We proceed to secure the properties required for the base case of the induction. The property $1$-harmony is trivially satisfied. We secure $1$-zero by adding to $\Phi_{m,n}$ the conjunct
\begin{align}
& \tag{$\Phi_1$}
\forall x_1 (\kint{1}(x_1) \rightarrow (\kzero{1}(x_1) \leftrightarrow \bigwedge_{i=0}^{n-1} \neg p_i(x_1))).
\label{eq:phi1}
\end{align}

\smallskip

\noindent
Thus,
if $a$ is a 1-integer, $\fA \models \kzero{1}[b] \Leftrightarrow \val{1}(b) = 0$.
To do the same for $1$-predecessor and $1$-equality,
we proceed as follows. Letting $L = 2m-1$,
we add to $\sigma_{m,n}$ an ($\ell+1$)-ary predicate, $p^\ell_i$, for all $i$ ($0 \leq i < n$) and all $\ell$ ($0 \leq \ell \leq L$), and we add to $\Phi_{m,n}$ the corresponding pair of conjuncts
\begin{equation}
\begin{split}
&\!\! \bigwedge_{i=0}^{n-1} \ \bigwedge_{\ell=0}^{L} \forall x_1 (\kint{1}(x_1) \wedge \pm p_i(x_1) \rightarrow\\
& \qquad \qquad  \quad  \forall x_2  \cdots \forall x_{\ell+1} \pm p_i^{\ell}(x_1, \dots, x_{\ell+1})).
\end{split}
\tag{$\Phi_2$}
\label{eq:phi2}
\end{equation}
Note that this really is a \textit{pair} of formulas: all occurrences of the
$\pm$ sign must be resolved in the same way.

Then, for any 1-integer $b$ and any $\ell$-tuple $\bar{c}$ from $A$,
\begin{equation}
\fA \models p^\ell_i[b, \bar{c}] \Leftrightarrow \fA \models p_i[b].
\label{eq:pid}
\end{equation}
In effect, the conjuncts~\eqref{eq:phi2} append  semantically inert arguments to each of the predicates $p_i$. This technique will be helpful at several points in the sequel, and we employ the convention that a superscript $\ell$ on a predicate letter indicates that the corresponding
undecorated predicate has $\ell$ semantically inert arguments {\em appended} to its primary arguments. Note that $p^0_i$ is simply equivalent to $p_i$.

Now we can secure the property $1$-equality. For all $\ell$ ($0 \leq \ell < 2m-2$), let
$\varepsilon_{1,\ell}(x_1, \dots, x_{\ell+2})$ abbreviate the formula: 
\begin{equation*}
\bigwedge_{i=0}^{n-1} (p^{\ell+1}_i(x_1, \dots, x_{\ell+2}) \leftrightarrow p_i(x_{\ell+2})).
\end{equation*}
We see from~\eqref{eq:pid} that
$\varepsilon_{1,\ell}(x_1, \dots, x_{\ell+2})$ in effect states that (for $x_1$ and $x_{\ell+2}$ 1-integers) the values of $x_1$ and $x_{\ell +2}$ are identical. We therefore add to $\Phi_{m,n}$ the conjuncts
\medskip
\begin{multline}
\smash{\bigwedge_{\ell=0}^{2m-2}} \forall x_1 (\kint{1}(x_1) \rightarrow \\
\hspace{-3cm} \forall x_2 \cdots \forall x_{\ell+2} (\kint{1}(x_{\ell+2}) \rightarrow\\
\kdeq{1}{\ell}(x_1, \dots, x_{\ell+2}) \leftrightarrow
\varepsilon_{1,\ell}(x_1, \dots, x_{\ell+2}))).
\tag{$\Phi_3$}
\label{eq:phi3}
\end{multline}
Thus, for any $1$-integers $b$, $b'$ in $\fA$ and any $\ell$-tuple $\bar{c}$ from $A$ ($0 \leq \ell < 2m-2$), $\fA \models \kdeq{1}{\ell}[b, \bar{c}, b'] \Leftrightarrow \val{1}(b) = \val{1}(b')$.

Turning to the property $1$-predecessor, assume for the moment that $m >1$, so that
the predicates $\kdpred{1}{0}$ and $\kdpred{1}{1}$ are both in $\sigma_{m,n}$.
For $0 \leq \ell \leq 1$, let $\pi_{1,\ell}(x_1, \dots, x_{\ell+2})$ abbreviate the formula
\bigskip
\begin{equation*}
\begin{split}
& \smash{\bigwedge_{i=0}^{n-1}} (
[\smash{\bigwedge_{j=0}^{i-1}} \neg p^{\ell+1}_j(x_1, \dots, x_{\ell+2}) \rightarrow
%& \qquad  \qquad  \qquad
   (p^{\ell+1}_i(x_1, \dots, x_{\ell+2}) \leftrightarrow \neg p_i(x_{\ell+2}))]
\wedge\\
& \qquad \qquad
[{\bigvee_{j=0}^{i-1}} p^{\ell+1}_j(x_1, \dots, x_{\ell+2}) \rightarrow
%& \qquad \qquad  \qquad
   (p^{\ell+1}_i(x_1, \dots, x_{\ell+2}) \leftrightarrow p_i(x_{\ell+2}))]).
\end{split}
\end{equation*}
From our preliminary remarks on the canonical representations of numbers by bit-strings, we see that
$\pi_{1,\ell}(x_1, \dots, x_{\ell+2})$ codes the statement that (for $x_1$ and $x_{\ell+2}$ 1-integers) the value of $x_{\ell +2}$ is one less than that of $x_1 \mod 2^n$ (empty disjunction is defined as false and empty conjunction as true).
We then add to $\Phi_{m,n}$ the conjuncts
\smallskip
\begin{multline}
\smash{\bigwedge_{\ell = 0}^{1}} \forall x_1 (\kint{1}(x_1) \rightarrow \\
\hspace{-4cm} \forall x_{2} \cdots \forall x_{\ell+2} (\kint{1}(x_{\ell +2}) \rightarrow\\
(\kdpred{1}{\ell}(x_1, \dots, x_{\ell+2}) \leftrightarrow \pi_{1,\ell}(x_1, \dots, x_{\ell+2})))),
\tag{$\Phi_{4}$}
\label{eq:phi4}
\end{multline}
securing the property $1$-predecessor, as required.

If, on the other hand,  $m=1$, we proceed in the same way, except that we add only the conjunct of~\eqref{eq:phi4} with index $\ell = 0$, i.e. the formula
\begin{multline}
\forall x_1 (\kint{1}(x_1) \rightarrow \forall x_{2} (\kint{1}(x_{2}) \rightarrow
(\kdpred{1}{0}(x_1, x_{2}) \leftrightarrow \pi_{1,0}(x_1, x_{2})))).
\tag{$\Phi_{4}'$}
\label{eq:phi4Prime}
\end{multline}
This suffices to satisfy the property 1-predecessor without resorting
to any predicates outside $\sigma_{1,n}$.

Finally, to secure $1$-covering, we add to $\Phi_{m,n}$ the conjuncts
\begin{align}
\tag{$\Phi_5$} & \exists x_1 (\kint{1}(x_1) \wedge \kzero{1}(x_1))
\label{eq:phi5}\\
\tag{$\Phi_6$} & \forall x_1 (\kint{1}(x_1) \rightarrow \exists x_2 (\kint{1}(x_2) \wedge \kdpred{1}{0}(x_1,x_2)))
\label{eq:phi6}.
\end{align}
Observe that~\eqref{eq:phi6} features only $\kdpred{1}{0}$, and not $\kdpred{1}{1}$, and so does not stray outside
$\sigma_{m,n}$, even when $m=1$.

\bigskip
\noindent
{\em Inductive case: } This case arises only if $m \geq2$.
Assume that, for some $k < m$,  $\val{k}: \kint{k}^\fA \rightarrow
[0,\tower{k}{n} -1]$  satisfies the properties of $k$-harmony,
$k$-zero, $k$-predecessor, $k$-covering and $k$-equality.
We show, by adding appropriate conjuncts to $\Phi_{m,n}$, that these properties hold with $k$ replaced by $k+1$.

For ($k+1$)-harmony, we add to $\Phi_{m,n}$ the following pair of conjuncts:
\begin{multline}
\forall x_1(\kint{k}(x_1) \rightarrow\\
\hspace{-5cm} \forall x_2(\kint{k+1}(x_2) \wedge \pm \kin{k}(x_1,x_2) \rightarrow \\
\forall x_3(\kint{k}(x_3) \myWedge \kdeq{k}{1}(x_1,x_2,x_3) \myRightarrow \pm \kout{k}(x_2, x_3)))).
\tag{$\Phi_{7}$}
\label{eq:phi7}
\end{multline}
If $a$, $a'$ are $k$-integers  such that $\val{k}(a) = \val{k}(a')$, and $b$ is any ($k+1$)-integer,  then, by $k$-equality (IH), $\fA \models \kdeq{k}{1}[a,b,a']$, whence~\eqref{eq:phi7} evidently secures
($k+1$)-harmony.

We remind ourselves at this point of the role of ($k+1$)-harmony in the subsequent argument, and, in particular,
on its relationship to $k$-covering. Let $b$ be a ($k+1$)-integer, and recall that $\val{k+1}(b)$ is defined by $b$'s satisfaction of the predicates $\kin{k}$ in relation to the various $k$-integers in $\fA$. By $k$-covering (IH), for all $i$ ($0 \leq i < \tower{k}{n}$), there {\em is} a $k$-integer $a$
with $\val{k}(a) = i$; and by ($k+1$)-harmony (just established), all such $k$-integers $a$ agree on
what the $i$th bit in $\val{k+1}(b)$ should be.

To secure ($k+1$)-zero, we add to $\Phi_{m,n}$ the conjunct
\begin{equation}
 \forall x_1 (\kint{k+1}(x_1) \rightarrow
 (\kzero{k+1}(x_1) \leftrightarrow
 \forall x_2 (\kint{k}(x_2) \rightarrow \neg \kout{k}(x_1,x_2)))).
\tag{$\Phi_{8}$}
\label{eq:phi8}
\end{equation}
From ($k+1$)-harmony and~\eqref{eq:phi8} we see that, for all ($k+1$)-integers $b$, $\fA \models \kzero{k+1}[b] \Leftrightarrow (\val{k+1}(b) = 0)$. For if there were any $k$-integer $a$, such that $\fA \models \kin{k}[a,b]$, then we would have
$\fA \models \kout{k}[b,a]$.

Establishing the property ($k+1$)-predecessor is more involved.
We add to $\sigma_{m,n}$ binary predicates $\kinStar{k}$, $\koutStar{k}$. The idea is that, for any $k$-integer $a$ and any ($k+1$)-integer $b$:
\begin{align}
\fA \models \kinStar{k}[a,b] & \Leftrightarrow \label{eq:instar}\\
 & \text{(for any $k$-integer $a'$,}
 \text{$\val{k}(a') < \val{k}(a) \Rightarrow \fA \not \models \kin{k}[a',b]$);}
\nonumber \\
\fA \models \kinStar{k}[a,b] & \Leftrightarrow \fA \models \koutStar{k}[b,a].
\label{eq:outstar}
\end{align}
Condition~\eqref{eq:instar} allows us to read $\kinStar{k}(x_1,x_2)$ as ``all the bits in the value of
the ($k+1$)-integer $x_2$ whose index is less than the value of
the $k$-integer $x_1$ are zero.'' Condition~\eqref{eq:outstar} is somewhat analogous to ($k+1$)-harmony.

Securing Condition~\eqref{eq:outstar} is easy. We add to $\Phi_{m,n}$ the pair of conjuncts
\begin{multline}
\forall x_1(\kint{k}(x_1) \rightarrow\\
\hspace{-5cm} \forall x_2(\kint{k+1}(x_2) \wedge \pm \kinStar{k}(x_1,x_2) \rightarrow \\
\forall x_3(\kint{k}(x_3) \myWedge \kdeq{k}{1}(x_1,x_2,x_3) \rightarrow
\pm \koutStar{k}(x_2, x_3)))). \hspace{1cm}
\label{eq:phi9}
\tag{$\Phi_{9}$}
\end{multline}
For let $a$ be a $k$-integer and $b$ a ($k+1$)-integer. By the property $k$-equality (IH),
$\fA \models \kdeq{k}{1}[a,b,a]$, whence~\eqref{eq:outstar} follows.

Securing Condition~\eqref{eq:instar} is harder. We first add to $\sigma_{m,n}$
a binary predicate $\kzero{k}^1$, which appends one semantically
inert argument to the unary predicate $\kzero{k}$. That is, we add to $\Phi_{m,n}$ the pair of conjuncts
\begin{equation}
\begin{split}
& \forall x_1 (\kint{k}(x_1) \hspace{-0.5mm} \wedge \hspace{-0.5mm}\pm  \kzero{k}(x_1) \myRightarrow
\forall x_2 \hspace{-0.5mm} \pm \hspace{-0.5mm} \kzero{k}^1(x_1,x_2)).\hspace{1cm}
\end{split} \hspace{-0.75cm}
\tag{$\Phi_{10}$}
\label{eq:phi10}
\end{equation}
We can then secure ~\eqref{eq:instar} by adding to $\Phi_{m,n}$ the conjunct
\begin{multline}
\forall x_1 ( \kint{k}(x_1) \rightarrow\\
\hspace{-2cm} \forall x_2 ( \kint{k+1}(x_2) \rightarrow
(\kinStar{k}(x_1, x_2) \leftrightarrow
(\kzero{k}^1(x_1,x_2) \vee \\
\forall x_3(\kint{k}(x_3) \wedge \kdpred{k}{1}(x_1, x_2, x_3) \rightarrow\\
(\koutStar{k}(x_2, x_3) \wedge \neg \kout{k}(x_2, x_3))))))). \hspace{1cm}  
\label{eq:phi11}
\tag{$\Phi_{11}$}
\end{multline}

\begin{figure}
\begin{center}
\resizebox{!}{2.8cm}{
\begin{tikzpicture}
\draw[->] (0,0) -- node[left]{$\kinStar{k}$} (0.93,0.93);
\draw[->] (1,0.75) -- node[below]{$\neg \kout{k}, \ \ \koutStar{k}$} (3.4,0.75);
\draw (1,0.65) -- (1,0.85);
\draw[dotted] (1,1) -- (3.5,2.5) -- (3.5, 1.5) -- cycle;
\draw (2.6,1.8) node[below]{$\neg \kout{k}$};

\draw (5,2.25) node[below]{$\val{k}(a') < \val{k}(a^*)$};
\draw (5.5,1.25) node[below]{$\val{k}(a^*) = \val{k}(a)-1$};

\draw[->] (0.39,-0.39) -- (1,0.2) -- node[below]{$\kdpred{k}{1}$}(3.5, 0.2);
\draw[thin] (0.32,-0.32) -- (0.46,-0.46);

\draw (-0.25,0) node {$a$};
\draw (1,1.3) node {$b$};
\draw (3.6,0.7) node {$a^*$};

\filldraw[fill=black!30] (0,0) circle (0.1);
\filldraw[fill=black] (1,1) circle (0.1);
\filldraw[fill=black!30] (3.5,1) circle (0.1);
\filldraw[fill=black!30] (3.5,1.5) circle (0.1);
\filldraw[fill=black!30] (3.5,2.5) circle (0.1);
\end{tikzpicture}
}
\end{center}
\caption{Fixing the interpretation of $\kinStar{k}$.}
\label{fig:kinstar}
\end{figure}

To see this, we perform a subsidiary
induction on the quantity $\val{k}(a)$.
Let $a$ be any $k$-integer and $b$ any ($k+1$)-integer.
For the base case, suppose $\val{k}(a) =0$. Then $\fA \models \kzero{k}[a]$ by the property
$k$-zero (IH), whence
$\fA \models \kzero{k}^1[a,b]$ by~\eqref{eq:phi10}, whence $\fA \models \kinStar{k}[a,b]$ by~\eqref{eq:phi11}.
For the inductive step, suppose that $\val{k}(a) > 0$; thus, by
$k$-zero again,
$\fA \not \models \kzero{k}[a]$.
Assume first that $\fA \models \kinStar{k}[a,b]$.
By $k$-covering (IH), we may pick some $k$-integer $a^*$ with $\val{k}(a^*) = \val{k}(a)-1$.
By $k$-predecessor (IH), setting $\ell = 1$, $\fA \models \kdpred{k}{1}[a,b,a^*]$, whence, taking $x_1$, $x_2$ and $x_3$
in~\eqref{eq:phi11} to be $a$, $b$ and $a^*$, respectively,
$\fA \models \koutStar{k}[b, a^*]$ and $\fA \not \models \kout{k}[b,a^*]$.
The situation is illustrated in Fig.~\ref{fig:kinstar}. Applying the subsidiary inductive hypothesis, it follows from~\eqref{eq:instar} and~\eqref{eq:outstar}, with $a$ replaced by $a^*$, that
for any $k$-integer $a'$ with $\val{k}(a') < \val{k}(a^*)$, $\fA \not \models \kin{k}[a',b]$.
Moreover, by ($k+1$)-harmony (just established), $\fA \not \models  \kout{k}[b,a^*]$
implies that,
for any $k$-integer $a'$ with $\val{k}(a') = \val{k}(a^*)$,  $\fA \not \models \kin{k}[a',b]$.
Thus, for any $k$-integer $a'$,
$\val{k}(a') < \val{k}(a) \Rightarrow \fA \not \models \kin{k}[a',b]$.
Conversely, suppose that $\fA \not \models \kinStar{k}[a,b]$.
Then, from~\eqref{eq:phi11}, there exists some  $k$-integer
$a^*$ such that $\kdpred{k}{1}[a,b,a^*]$, but either
$\fA \not \models \koutStar{k}[b, a^*]$ or $\fA \models \kout{k}[b,a^*]$.
By $k$-predecessor (IH), again setting $\ell =1$, $\val{k}(a^*) = \val{k}(a)-1$, and hence,
applying the subsidiary inductive hypothesis, \eqref{eq:instar} and~\eqref{eq:outstar} ensure that,
if $\fA \not \models \koutStar{k}[b, a^*]$,
then, for some $k$-integer $a'$ with $\val{k}(a') < \val{k}(a^*) < \val{k}(a)$, $\fA \models \kin{k}[a',b]$.
On the other hand, by~$k$-equality and ($k+1$)-harmony,  $\fA \models \kout{k}[b,a^*]$ implies
$\fA \models \kin{k}[a^*,b]$. Either way, there exists a $k$-integer $a'$ such that
$\val{k}(a') < \val{k}(a)$, but $\fA \models \kin{k}[a',b]$.
This completes the (subsidiary) induction, and establishes~\eqref{eq:instar}.

Having fixed the interpretation of $\kinStar{k}$, we proceed to secure the property \linebreak
($k+1$)-predecessor.
Assume first that $k+1 < m$. We add to $\sigma_{m,n}$
the predicates $\kin{k}^{2}$, $\kin{k}^{3}$, $\kinStar{k}^{2}$, $\kinStar{k}^{3}$ and we add to $\Phi_{m,n}$ the conjuncts

\begin{multline}
\smash{\bigwedge_{\ell=0}^{1}} \forall x_1 (\kint{k}(x_1) \rightarrow\\
\hspace{-3cm}
\forall x_2 (\kint{k+1}(x_2) \wedge \pm \kin{k}(x_1,x_2) \rightarrow\\
\forall x_3 \cdots \forall x_{\ell+4} \pm \kin{k}^{\ell+2}(x_1, \dots, x_{\ell+4})))
\tag{$\Phi_{12}$}
\label{eq:phi12}
\end{multline}
\begin{multline}
\smash{ \bigwedge_{\ell=0}^{1}} \forall x_1 (\kint{k}(x_1) \rightarrow\\
\hspace{-3cm}
\forall x_2 (\kint{k+1}(x_2) \wedge \pm \kinStar{k}(x_1,x_2) \rightarrow\\
\forall x_3 \cdots \forall x_{\ell+4} \pm \kinStar{k}^{\ell+2}(x_1, \dots, x_{\ell+4})),
\tag{$\Phi_{13}$}
\label{eq:phi13}
\end{multline}
fixing these predicates to be the result of adding either 2 or 3 semantically inert arguments to $\kin{k}$ and $\kinStar{k}$, as indicated by the superscripts.

Now add to $\sigma_{m,n}$ the ternary predicate $\kdpredDig{k+1}{0}$ and quaternary
predicate $\kdpredDig{k+1}{1}$ and let $\rho_{k,\ell}(x_1, \dots, x_{\ell+4})$ abbreviate the formula
\begin{equation*}
\begin{split}
& [\kinStar{k}^{\ell+2}(x_1, \dots, x_{\ell+4}) \rightarrow
 (\kin{k}^{\ell+2}(x_1, \dots, x_{\ell+4}) \leftrightarrow \neg
\kout{k}(x_{\ell+3}, x_{\ell+4}))]
\ \wedge \\
& [\neg \kinStar{k}^{\ell+2}(x_1, \dots, x_{\ell+4}) \rightarrow
 (\kin{k}^{\ell+2}(x_1, \dots, x_{\ell+4}) \leftrightarrow
  \kout{k}(x_{\ell+3}, x_{\ell+4}))],
\end{split}
\end{equation*}
where $0 \leq \ell \leq 1$. Suppose that $a$, $a'$ are $k$-integers, $b$, $b'$ ($k+1$)-integers and $\bar{c}$ an $\ell$-tuple of elements. From~\eqref{eq:phi12}, $\fA \models \kin{k}^{\ell+4}[a,b,\bar{c},b',a']  \Leftrightarrow \fA \models \kin{k}[a,b]$, and
from~\eqref{eq:phi13},
$\fA \models \kinStar{k}^{\ell+2}[a,b,\bar{c},b',a']  \Leftrightarrow \fA \models \kinStar{k}[a,b]$.
Furthermore, by ($k+1$)-harmony,
$\fA \models \kout{k}[b',a']  \Leftrightarrow \fA \models \kin{k}[a',b']$.
Hence,
from our preliminary remarks on the canonical representations of numbers by bit-strings,
$\fA \models \rho_{k,\ell}[a,b,\bar{c},b',a']$ just in case the $\val{k}(a')$-th digit in the encoding of $\val{k+1}(b')$ is
the same as the $\val{k}(a)$-th digit in the encoding of $\val{k+1}(b) -1$,
modulo $\tower{k+1}{n}$.

We then add to $\Phi_{m,n}$ the conjunct
\begin{multline}
\smash{\bigwedge_{\ell=0}^1}\forall x_1 (\kint{k}(x_1) \rightarrow\\
\hspace{-5.7cm}
\forall x_2 (\kint{k+1}(x_2) \rightarrow\\
\hspace{-3.6cm} \forall x_{3} \cdots \forall x_{\ell+3} (\kint{k+1}(x_{\ell+3}) \rightarrow \\
\hspace{-0.5cm}
\forall x_{\ell+4} (\kint{k}(x_{\ell+4}) \wedge \kdeq{k}{\ell+2}(x_1,\dots,x_{\ell+4}) \rightarrow \\
(\kdpredDig{k+1}{\ell}(x_2, \dots, x_{\ell+4}) \leftrightarrow
\rho_{k,\ell}(x_1, \dots, x_{\ell+4}))))).
\tag{$\Phi_{14}$}
\label{eq:phi14}
\end{multline}
This formula is illustrated in the left-hand diagram of Fig.~\ref{fig:eqDig} in the case $\ell =1$: here,
 $\rho_{k,1}$ holds of the tuple $a,b,c,b',a$, and $\kdpredDig{k+1}{1}$ of the tuple $b,c,b',a$,
 just in case the $\val{k}(a)$th digit of $\val{k+1}(b')$ agrees with the $\val{k}(a)$th digit of $\val{k+1}(b)-1$. (Note that the single element
$a$ is depicted twice in this diagram.)
Suppose $a$ is a $k$-integer, $b$, $b'$ are ($k+1$)-integers in $\fA$,
and $\bar{c}$ is any $\ell$-tuple from $A$ with $0 \leq \ell \leq 1$.
By $k$-equality (IH), $\fA \models \kdeq{k}{\ell+2}[a,b,\bar{c}, b',a]$, and from the properties
of $\rho_{k,\ell}$ just established (setting $a'= a$),
$\fA \models \kdpredDig{k+1}{\ell}[b,\bar{c},b',a]$ just in case
the $\val{k}(a)$th digit of $\val{k+1}(b')$ is equal to the $\val{k}(a)$th digit of
$\val{k+1}(b) -1$, modulo $\tower{k+1}{n}$.

To establish ($k+1$)-predecessor, therefore, we add to $\Phi_{m,n}$ the conjuncts
\begin{multline}
\smash{\bigwedge_{\ell=0}^1}\forall x_1 (\kint{k+1}(x_1) \rightarrow
\forall x_2 \cdots \forall x_{\ell+2} (\kint{k+1}(x_{\ell+2}) \rightarrow\\
\hspace{-3cm} (\kdpred{k+1}{\ell}(x_1, \dots, x_{\ell+2}) \leftrightarrow\\ 
\forall x_{\ell+3} (\kint{k}(x_{\ell+3}) \rightarrow
\kdpredDig{k+1}{\ell}(x_1, \dots, x_{\ell+3}))))). \hspace{0.4cm}
\tag{$\Phi_{15}$}
\label{eq:phi15}
\end{multline}
From~\eqref{eq:phi15},
$\fA \models \kdpred{k+1}{\ell}[b,\bar{c},b']$ just in case
each digit of $\val{k+1}(b')$ is equal to the corresponding digit of
$\val{k+1}(b) -1$, modulo $\tower{k+1}{n}$.

If, on the other hand, $k+1 = m$, we proceed as above, but we add to $\sigma_{m,n}$ only the predicates $\kin{k}^2$, $\kinStar{k}^2$, $\kdpredDig{k+1}{0}$ (not $\kin{k}^3$, $\kinStar{k}^3$ or $\kdpredDig{k+1}{1}$), and we add to $\Phi_{m,n}$ only those conjuncts
of~\eqref{eq:phi12}--\eqref{eq:phi15} with $\ell=0$ (not with $\ell=1$). This suffices for ($k+1$)-predecessor in the case $k+1 = m$, and does
not require the use of any predicates outside $\sigma_{m,n}$. Note in particular that the conjunct
$\ell = 1$ in~\eqref{eq:phi12}--\eqref{eq:phi14} requires predicates of arity 5, which are not available when $m=2$. 

To establish the property
($k+1$)-covering, we add to $\Phi_{m,n}$ the conjuncts
\begin{align}
& \exists x_1 (\kint{k+1}(x_1) \wedge \kzero{k+1}(x_1))
\tag{$\Phi_{16}$}
\label{eq:phi16}\\
& \forall x_1 (\kint{k+1}(x_1) \rightarrow
 \exists x_2 (\kint{k+1}(x_2) \wedge \kdpred{k+1}{0}(x_1,x_2))).
\tag{$\Phi_{17}$}
\label{eq:phi17}
\end{align}
Note that~\eqref{eq:phi17} features only $\kdpred{k+1}{0}$, and not $\kdpred{k+1}{1}$, so it is defined even when $k+1=m$

It remains only to establish ($k+1$)-equality.
Conceptually, this is rather easier than ($k+1$)-predecessor; however, we do need to consider larger  numbers of semantically inert variables. Let $L = 2(m-k-1)$. The property ($k+1$)-equality concerns the interpretation of the ($\ell+2$)-ary predicate $\kdeq{k+1}{\ell}$ 
for all $\ell$ ($0 \leq \ell \leq L$). Observe that, if $k=1$ (first inductive step), then $L=2m-4$, and if
$k= m-1$ (last inductive step), then $L=0$. Thus, in the sequel, we always have $L \leq 2m-4$. (Remember that the inductive case
is encountered only if $m \geq 2$.)
To ease the pain of reading, we split the task into three stages.

For the
first stage, for all $\ell$ ($0  \leq \ell \leq L$), add to $\sigma_{m,n}$
an ($\ell+2$)-ary predicate $\kin{k}^\ell$, and add to $\Phi_{m,n}$ the conjuncts
\begin{multline}
\smash{\bigwedge_{\ell=0}^{L}}\forall x_1 (\kint{k}(x_1) \rightarrow \\
\hspace{-3cm}
\forall x_2 (\kint{k+1}(x_2) \wedge \pm \kin{k}(x_1,x_2) \rightarrow\\
\forall x_3 \cdots \forall x_{\ell+2} \pm \kin{k}^{\ell}(x_1, x_2, \dots, x_{\ell+1}, x_{\ell+2}))),
\tag{$\Phi_{18}$}
\label{eq:phi18}
\end{multline}
thus fixing $\kin{k}^\ell$ to be the result of adding $\ell$  semantically inert arguments to  $\kin{k}^\ell$.
(For $\ell \leq 3$, this repeats the work of~\eqref{eq:phi12}, but no matter.)

In the second stage, for all $\ell$ ($0  \leq \ell \leq L$), add to $\sigma_{m,n}$ an
($\ell+3$)-ary predicate $\kdeqDig{k+1}{\ell}$, and add to $\Phi_{m,n}$ the conjuncts

\smallskip

\begin{multline}
\smash{\bigwedge_{\ell=0}^{L}}
\forall x_1 (\kint{k}(x_1) \rightarrow
\forall x_2 (\kint{k+1}(x_{2}) \rightarrow\\
\hspace{-3.5cm}
\forall x_3 \cdots \forall x_{\ell+3} (\kint{k+1}(x_{\ell+3}) \rightarrow \\
\hspace{-0cm}
\forall x_{\ell+4} (\kint{k}(x_{\ell+4}) \wedge \kdeq{k}{\ell+2}(x_1, \dots, x_{\ell+4}) \rightarrow\\
(\kdeqDig{k+1}{\ell}(x_2, \ldots, x_{\ell+4})  \leftrightarrow
\eta_{k,\ell+2}(x_1, \dots, x_{\ell+4})))))),
\tag{$\Phi_{19}$}
\label{eq:phi19}
\end{multline}
where $\eta_{k,\ell+2}(x_1, \dots, x_{\ell+4})$ is the formula: $\kin{k}^{\ell+2}(x_1, \dots, x_{\ell+4}) \leftrightarrow \kout{k}(x_{\ell+3},x_{\ell+4})$.

Let $b$, $b'$ be ($k+1$)-integers in $\fA$, $a$ a $k$-integer in $\fA$, and $\bar{c}$ any $\ell$-tuple
from $A$.
We claim that $\fA \models \kdeqDig{k+1}{\ell}[b, \bar{c}, b',a]$ just in case
$\val{k+1}(b)$ and $\val{k+1}(b')$ agree on their $\val{k}(a)$th bit.
For, by $k$-equality (IH),
$\fA \models \kdeq{k}{\ell+2}[a,b, \bar{c}, b',a]$.
Hence, by~\eqref{eq:phi19}
$\fA \models \kdeqDig{k+1}{\ell}[b, \bar{c}, b',a]$ holds just in case
$\fA \models \eta_{k,\ell+2}[a, b, \bar{c}, b', a]$.
But, by~\eqref{eq:phi18},
$\fA \models \kin{k}^{\ell+2}[a, b, \bar{c}, b', a]$ if and only if
$\fA \models \kin{k}[a, b]$, i.e.~if and only if the $\val{k}(a)$th bit of $\val{k+1}(b)$ is 1.
That is,
$\fA \models \kdeqDig{k+1}{\ell}[b, \bar{c}, b',a]$ is equivalent to the statement that
$\fA \models \kin{k}[a,b]$ if and only if $\fA \models \kout{k}[b', a]$.
The situation is illustrated (for the case where
$\fA \models \kdeqDig{k+1}{\ell}[b, \bar{c}, b',a]$ holds) in the right-hand
diagram of Fig.~\ref{fig:eqDig}, where all polarity alternatives
$\pm$ are assumed to be resolved in the same way.
But, by ($k+1$)-harmony, $\fA \models \kout{k}[b', a]$ if and only if
$\fA \models \kin{k}[a, b']$.
This establishes the claim.
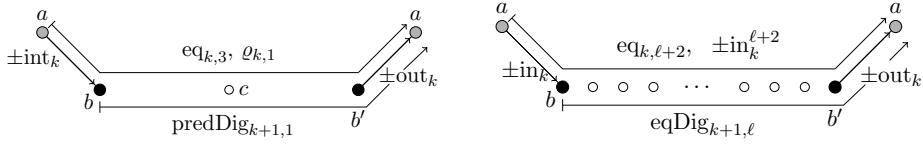
\begin{figure}
\begin{center}
\resizebox{!}{1.9cm}{
	\begin{tikzpicture}
	\draw (0.5, 2.8) node  {$a$};
	\draw (0.5,2.5) -- (1.5,1.5);
	\draw[->] (0.5,2.5) -- node[left] {$\pm \kint{k}$} (1.4,1.6);
	
	\draw (7, 2.8) node  {$a$};
	\draw (6,1.5) -- (7,2.5);
	\draw[->] (6,1.5) -- node[right, near start] {$\pm \kout{k}$} (6.9,2.4);
	
	\draw (1.3, 1.3) node  {$b$};
	\draw[] (0.66,2.64) -- (1.5,1.8) -- node[above] {$\kdeq{k}{3}$, $\rho_{k,1}$} (6, 1.8);
	\draw[->] (6, 1.8) -- (6.84, 2.64);
	\draw[thin] (0.6,2.58) -- (0.72,2.70);
	
	\draw (4, 1.5) node  {$c$};
	
	\draw (6, 1) node  {$b'$};
	\draw[] (1.5,1.2) -- node[below] {$\kdpredDig{k+1}{1}$} (6.14, 1.2) -- (6.54, 1.6);
	\draw[->] (6.79, 1.85) --  (7.20, 2.24);
	\draw[thin] (1.5,1.1) -- (1.5,1.3);
	
	\filldraw[fill=black!30] (0.5,2.5) circle (0.1);
	\filldraw[fill=black] (1.5,1.5) circle (0.1);
	
	\filldraw[fill=white] (3.75,1.5) circle (0.075);
	
	\filldraw[fill=black] (6,1.5) circle (0.1);
	\filldraw[fill=black!30] (7,2.5) circle (0.1);
	\end{tikzpicture}}
\hspace{2mm}
\resizebox{!}{2cm}{
\begin{tikzpicture}
\draw (0.5, 2.8) node  {$a$};
\draw (0.5,2.5) -- (1.5,1.5);
\draw[->] (0.5,2.5) -- node[below] {\hspace{-2.5mm} $\pm \kin{k}$} (1.4,1.6);

\draw (7, 2.8) node  {$a$};
\draw (6,1.5) -- (7,2.5);
\draw[->] (6,1.5) -- node[below right] {\hspace{-3mm} $\pm \kout{k}$} (6.9,2.4);

\draw (1.3, 1.3) node  {$b$};
\draw[] (0.66,2.64) -- (1.5,1.8) -- node[above] {$\kdeq{k}{\ell+2},\ \ \pm \kin{k}^{\ell+2}$} (6, 1.8);
\draw[->] (6, 1.8) -- (6.84, 2.64);
\draw[thin] (0.6,2.58) -- (0.72,2.70);

\draw (6, 1) node  {$b'$};
\draw[] (1.5,1.2) -- node[below] {$\kdeqDig{k+1}{\ell}$} (6.14, 1.2) -- (6.5, 1.54);
\draw[->] (6.9, 1.94) -- (7.20, 2.24);
\draw[thin] (1.5,1.1) -- (1.5,1.3);
\draw (3.75, 1.5) node  {$\dots$};

\filldraw[fill=black!30] (0.5,2.5) circle (0.1);
\filldraw[fill=black] (1.5,1.5) circle (0.1);

\filldraw[fill=white] (2,1.5) circle (0.075);
\filldraw[fill=white] (2.5,1.5) circle (0.075);
\filldraw[fill=white] (3,1.5) circle (0.075);
\filldraw[fill=white] (4.5,1.5) circle (0.075);
\filldraw[fill=white] (5,1.5) circle (0.075);
\filldraw[fill=white] (5.5,1.5) circle (0.075);

\filldraw[fill=black] (6,1.5) circle (0.1);
\filldraw[fill=black!30] (7,2.5) circle (0.1);
\end{tikzpicture}
}
\end{center}
\caption{Fixing the interpretations of
	 $\kdpredDig{k+1}{1}$ (left) and $\kdeqDig{k+1}{\ell}$ (right).}
\label{fig:eqDig}
\end{figure}

In the third stage, we add to $\Phi_{m,n}$ the conjunct
\begin{multline}
\smash{\bigwedge_{\ell=0}^{L}} \forall x_1 (\kint{k+1}(x_1) \rightarrow\\
\hspace{-2cm} \forall x_2 \cdots \forall x_{\ell+2} (\kint{k+1}(x_{\ell+2}) \rightarrow \\
\hspace{-1cm}
(\kdeq{k+1}{\ell}(x_1, \dots, x_{\ell+2}) \leftrightarrow\\
\forall x_{\ell+3} (\kint{k}(x_{\ell+3}) \rightarrow \kdeqDig{k+1}{\ell}(x_1, \dots, x_{\ell+3}))))).
\tag{$\Phi_{20}$}
\label{eq:phi20}
\end{multline}
Given the properties of $\kdeqDig{k+1}{\ell}$ just established, this evidently secures ($k+1$)-equality, completing the induction.

\bigskip
We have remarked that, by $m$-covering, any model of $\Phi_{m,n}$ has cardinality at least $\tower{m}{n}$.
We claim that $\Phi_{m,n}$ is satisfiable.  Let
$A = A_1 \dot{\cup} \cdots \dot{\cup} A_m$, where $A_k = \{ \langle k, i\rangle\mid 0\leq i< \tower{k}{n}\} \rangle$. (That is, $A$ is the
disjoint union of the various sets of integers $[0,\tower{k}{n}-1]$.) Let $\kint{k}^\fA = A_k$ for all $k$ ($1 \leq k \leq m$), and interpret the
other predicates of $\sigma_{m,n}$ as described above. It is easily verified that $\fA \models \Phi_{m,n}$.

It remains only to check the number of variables featured in $\Phi_{m,n}$.
Consider first the conjuncts introduced in the base case. By inspection, \eqref{eq:phi1}--\eqref{eq:phi3} and~\eqref{eq:phi5}--\eqref{eq:phi6}
are in $\FL^{2m}$. For $m >1$,~\eqref{eq:phi4} is in $\FL^{3}$; but if $m=1$,
only the conjunct with  index $\ell = 0$ is present, which is in $\FL^2$. Either way, \eqref{eq:phi1}--~\eqref{eq:phi6} are in $\FL^{2m}$.
Consider now the conjuncts introduced in the inductive case. By inspection,
these feature only $\max(5,2m) \leq 2m$ variables. If, however,
$m=2$, then the inductive step only runs once, with $k+1 = m$, in which case
only those conjuncts of~\eqref{eq:phi12}--\eqref{eq:phi15} occur for which $\ell=0$,
which feature only 4 variables.
Either way, \eqref{eq:phi7}--\eqref{eq:phi20} are in $\FL^{2m}$.
\end{proof}

\begin{theorem}\label{thm:lower}
The satisfiability problem for $\FL^{2m}$ is $m$-\NExpTime-hard.\newline
Hence, the satisfiability problem for $\FL$ is non-elementary.
\end{theorem}
\begin{proof}
A {\em tiling system} is a quadruple $(C,c_0,H,V)$, where
$C$ is a non-empty, finite set, $c_0$ an element of $C$, and $H$, $V$  subsets of $C \times C$. 
For any $N \geq 2$, an $N \times N$ \textit{tiling} for $(C,c_0,V,H)$ is a function $f: [0, N-1] \times [0, N-1] \rightarrow C$ such that $f(0,0) = c_0$ and, for all $i$, $j$ ($0 \leq i,j < N$), $\langle f(i,j), f(i+1,j)\rangle \in H$ and $\langle f(i,j), f(i,j+1)\rangle \in V$, where arithmetic in subscripts is interpreted modulo $N$. Intuitively, we are to imagine an $N \times N$ grid (with toroidal wrap-around) whose squares are tiled with the tiles of the colours in $C$: the `bottom left' square is to be coloured $c_0$, the relation $H$ specifies which colours are allowed to go immediately `to the right of' which other colours, and the relation $V$ specifies which colours are allowed to go immediately `above' which other colours. For a fixed,
non-negative integer $m$, the $\tower{m}{n}$-sized tiling problem $(C,c_0,H,V)$ asks whether, for a positive
integer $n$, there exists a $\tower{m}{n} \times \tower{m}{n}$ tiling for $(C,c_0,V,H)$. It is well-known that,
for any positive $m$,
there exist tiling systems $(C,c_0,V,H)$ whose $\tower{m}{n}$-sized tiling problem is $m$-\NExpTime-hard.

Fixing some positive $m$, and given any tiling system $(C,c_0,V,H)$, we shall construct,
in time bounded by a
polynomial function of $n$, a formula $\Psi_n$ of $\FL^{2m}$ such that $\Psi_n$ is satisfiable if and only
if there exists a $\tower{m}{n} \times \tower{m}{n}$ tiling for $(C,c_0,V,H)$. It follows that 
the satisfiability problem for $\FL^{2m}$ is $m$-\NExpTime-hard.

The proof in the case $m=1$ is an easy adaptation of the standard proof of the \NExpTime-hardness proof for the two-variable fragment of first-order logic, and need not be rehearsed here; see,~e.g.
B\"{o}rger, Gr\"{a}del and Gurevich~\cite[pp.~253 ff.]{purdy:BGG97}. Hence, we may assume that $m >1$. We proceed  as with the construction of $\Phi_{m,n}$ in the proof of Theorem~\ref{thm:bigModels}, except that we slightly modify the final inductive step.
Specifically, we begin by supposing $\Psi_n$ to have all those conjuncts of $\Phi_{m,n}$ required to establish 
the existence of $k$-integers, with a valuation
function $\val{k}$ satisfying the properties $k$-harmony, $k$-zero, $k$-predecessor, $k$-covering, $k$-equality,
for all values of $k$ from 1 to $m-1$. We do not construct $m$-integers, but instead
objects we refer to as {\em vertices}, which are, in effect, pairs of $m$-integers. We establish the requisite properties
of vertices  by adding to
$\Psi_n$ further conjuncts as described below. 

Let $\vertex$ be a unary predicate, and $\kin{X}$, $\kin{Y}$, $\kout{X}$, $\kout{Y}$, binary predicates. 
If a structure $\fA$ is clear from
context, we call any element of $A$ satisfying $\vertex$ a {\em vertex}. Define
the function $\val{X}: \vertex^\fA \rightarrow [0,\tower{m}{n}-1]$ by setting $\val{X}(b)$,
for any vertex $b$, to be
the integer coded by the bit-string $s_{N-1}, \dots, s_0$ of length $N = \tower{m-1}{n}$
where, for all $i$ ($0 \leq i < N$),
\begin{equation*}
s_i =
\begin{cases}
1 \text{ if $\fA \models \kin{X}[a,b]$ for some  ($m-1$)-integer $a$ such that $\val{m-1}(a) = i$;}\\
0 \text{ otherwise.}
\end{cases}
\end{equation*}
Thus, we use the ($m-1$)-integers to represent positions in the 
horizontal coordinate of any vertex, employing the binary predicate $\kin{X}$ to encode
the bits at the positions in question. 
Similarly, we may use the ($m-1$)-integers to represent positions in the 
vertical coordinate of any vertex, employing the binary predicate $\kin{Y}$ to encode
the bits at the positions in question. Here again, we implicitly rely on
($m-1$)-covering to show that there {\em is} a ($m-1$)-integer having any value $i$ in the range
$[0,\tower{m-1}{n}-1]$. In the sequel, 
we establish  harmony-like properties showing that it does not matter \textit{which} such ($m-1$)-integer we choose.

We write $\bar{X} = Y$ and $\bar{Y} = X$, and we use the symbol
$D$ to stand for either of the letters $X$ or $Y$.
We shall ensure that any model of $\Psi_n$ has the following properties, for $D \in \{X,Y\}$.
\begin{description}
\item[$D$-{harmony}] For all vertices $b$ and all ($m-1$)-integers $a$, $a'$ 
 in $\fA$ such that $\val{m-1}(a)= \val{m-1}(a')$, 
\mbox{$\fA \models \kin{D}[a,b] \Leftrightarrow \fA \models \kin{D}[a',b]$}.
\item[$D$-{zero}] For all vertices $b$  in $\fA$,
$\fA \models \kzero{D}[a] \Leftrightarrow (\val{D}(a) = \val{D}(a) = 0)$.
\item[$D$-{equality}] 
For all vertices $a$, $c$ in $\fA$,
$\fA \models \keq{D}[a,c] \Leftrightarrow \val{D}(c) = \val{D}(a)$.
\item[$D$-{predecessor}] For all vertices $a$, $c$ in $\fA$,
$\fA \models \kpred{D}[a,c] \Leftrightarrow \\
\val{D}(c) = \val{D}(a)-1$ modulo $\tower{m}{n}$. 
\end{description}

To do so, we add to $\Psi_n$ the following conjuncts. The argumentation is in each case 
virtually identical to that given in the proof of Theorem~\ref{thm:bigModels}.
For $X$- and $Y$-harmony, we add to our signature binary predicates $\kout{X}$, $\kout{Y}$,
and add to $\Psi_n$ the following conjuncts.
\begin{equation*}
\begin{split}
& \smash{\bigwedge_{D \in \{X,Y\}} }
\forall x_1(\kint{m-1}(x_1) \rightarrow \\
& \qquad \qquad \qquad \forall x_2(\vertex(x_2) \wedge \pm \kin{D}(x_1,x_2) \rightarrow \\
& \qquad \qquad \qquad \qquad \forall x_3(\kint{m-1}(x_3) \rightarrow\\
& \qquad \qquad \qquad \qquad \qquad 
 \kdeq{m-1}{1}(x_1,x_2,x_3) \rightarrow
 \pm \kout{D}(x_2, x_3)))).
\end{split}
\end{equation*}
Again, we observe in passing  that, if
$b$ is a vertex and $c$, $c'$ are ($m-1$)-integers with
$\val{k}(c) = \val{k}(c')$, then $\fA \models \kout{D}[b,c] \Leftrightarrow \fA \models \kout{D}[b,c']$.
At this point, we have established that, in any model $\fA$ of $\Psi_n$, any vertex $b$ has
well-defined coordinates $(\val{X}(b), \val{Y}(b))$.

For $X$- and $Y$-equality, add to $\Psi_n$ the following conjuncts.
\begin{align*}
\begin{split}
& \smash{\bigwedge_{D \in \{X,Y\}}}
\forall x_1 (\kint{m-1}(x_1) \rightarrow\\
& \qquad \qquad \qquad \forall x_2 (\vertex(x_{2}) \rightarrow \\
& \qquad \qquad \qquad \qquad \forall x_3 (\vertex(x_{3}) \rightarrow \\
& \qquad \qquad \qquad \qquad \qquad \forall x_{4} (\kint{m-1}(x_{4}) \wedge \kdeq{m-1}{2}(x_1, x_2, x_3, x_{4}) \rightarrow \\
& \qquad \qquad \qquad \qquad \qquad \qquad (\keqDig{D}(x_2, x_3, x_{4})  \leftrightarrow\\
& \qquad \qquad \qquad \qquad \qquad \qquad \qquad 
  (\kin{D}^{2}(x_1, x_2, x_3, x_{4}) \leftrightarrow \kout{m-1}(x_{3},x_{4}))))))).
\end{split}\\
\begin{split}
& \smash{\bigwedge_{D \in \{X,Y\}}}
\forall x_1 (\vertex(x_1) \rightarrow\\
& \qquad \qquad\qquad \forall x_2 (\vertex(x_{2}) \rightarrow \\
& \qquad \qquad\qquad \qquad (\keq{D}(x_1, x_{2}) \leftrightarrow\\
& \qquad \qquad\qquad \qquad  \qquad \forall x_{3} (\kint{m-1}(x_{3}) \rightarrow \keqDig{D}(x_1, x_2, x_{3}))))).
\end{split}
\end{align*}

For $X$- and $Y$-zero, add to $\Psi_n$ the following conjuncts.
\begin{equation*}
\bigwedge_{D \in \{X,Y\}}
   \forall x_1 (\vertex(x_1) \rightarrow \kzero{D}(x_1) \leftrightarrow \forall x_2 \neg \kout{D}(x_1,x_2)).
\end{equation*}
For $X$- and $Y$-predecessor, let $\rho_D(x_1, \dots, x_4)$ be the formula
\begin{equation*}
\begin{split}
& (\kinStar{D}^{2}(x_1, \dots, x_{4}) \rightarrow (\kin{D}^{2}(x_1, \dots, x_{4}) \leftrightarrow \neg 
\kout{D}(x_{3}, x_{4}))) 
\ \wedge \\
& \qquad  \qquad 
  (\neg \kinStar{D}^{2}(x_1, \dots, x_{4}) \rightarrow (\kin{D}^{2}(x_1, \dots, x_{4}) \leftrightarrow 
  \kout{k}(x_{3}, x_{4})))
\end{split}
\end{equation*}
for $D \in \{X,Y\}$, and add to $\Psi_n$ the following conjuncts.
\begin{align*}
\begin{split}
& \smash{\bigwedge_{D \in \{X,Y\}}} \forall x_1 (\kint{m-1}(x_1) \rightarrow \\
& \quad \quad \quad \quad \quad \quad \forall x_2 (\vertex(x_2) \rightarrow \\
& \quad \quad \quad \quad \quad \quad \quad \forall x_{3} (\vertex(x_{3}) \rightarrow \\
& \quad \quad \quad \quad \quad \quad \quad \quad \forall x_5 (\kint{m-1}(x_{4}) \wedge \kdeq{m-1}{2}(x_1,\dots,x_{\ell+4}) \rightarrow \\
& \quad \quad \quad \quad \quad \quad \quad \quad \qquad (\kpredDig{D}(x_2, x_3, x_{4}) \leftrightarrow \rho_D(x_1, \dots, x_{4})))))
\end{split}\\
\ \\
\begin{split}
& \smash{\bigwedge_{D \in \{X,Y\}}}\forall x_1 \forall x_{2} (\kpred{D}(x_1, x_{2}) \leftrightarrow\\
& \quad \quad \quad \qquad \qquad \qquad \forall x_{3} (\kint{m-1}(x_{3}) \rightarrow \kpredDig{D}(x_1, x_2, x_{3}))).
\end{split}
\end{align*}
At this point, we have established that, in any model $\fA$ of $\Psi_n$, and for $D \in \{X,Y\}$,
the properties
$D$-harmony, $D$-zero, $D$-predecessor and $D$-equality obtain.

The following conjuncts of $\Psi_n$ now establish that, for all pairs of integers $i$, $j$ in the range
$[0, \tower{m}{n}-1]$, there exists a vertex $a$ with coordinates $(x,y)$:
\begin{align*}
& \exists x_1 (\vertex(x_1) \wedge \kzero{X}(x_1) \wedge \kzero{Y}(x_1))\\
& \forall x_1 (\vertex(x_1) \rightarrow \exists x_2 (\vertex(x_2) \wedge \kpred{Y}(x_1,x_2) \wedge \keq{X}(x_1,x_2)))\\  
& \forall x_1 (\vertex(x_1) \rightarrow \exists x_2 (\vertex(x_2) \wedge \kpred{X}(x_1,x_2) \wedge \keq{Y}(x_1,x_2))).
\end{align*}

Treating the colours in $C$ as unary predicates,
the following conjuncts assign, to each vertex, $a$, a unique 
colour, namely, the colour which $a$ satisfies.  
\begin{align*}
& \forall x_1 (\vertex(x_1) \rightarrow \bigvee_{c \in C} c(x_1)) \\
& \bigwedge_{c, d \in C}^{c \neq d} \forall x_1 (\vertex(x_1) \rightarrow \neg (c(x_1) \wedge d(x_1)))
\end{align*}
Note that there is no requirement that vertices be uniquely defined by their $X$- and $Y$-coordinates.
Nevertheless, we obtain a well-defined encoding of a grid-colouring by securing the property
\begin{description}
\item[{chromatic harmony}] For all vertices $b$ and $b'$ such that
$\val{X}(b) = \val{X}(b)$ and $\val{Y}(b) = \val{Y}(b)$, and for all
colours $c \in C$, $\fA \models c[b] \Leftrightarrow \fA \models c[b']$.
\end{description}
To do so, we add to $\Psi_n$ the conjunct
\begin{align*}
&\forall x_1 (\vertex(x_1) \wedge c(x_1) \rightarrow \forall x_2 (\vertex(x_2) \wedge \keq{X}(x_1, x_2)
\wedge \keq{Y}(x_1, x_2) \rightarrow c(x_2))).
\end{align*}
Thus, any model $\fA $ of $\Psi_n$ defines a function
$f: [0,\tower{m}{n}-1]^2 \rightarrow C$. To ensure that $f$ is a tiling for the system  
$(C,c_0,H,V)$, we simply add to $\Psi_n$ the conjuncts 
\begin{align*}
& \forall x (\kzero{X}(x) \wedge \kzero{Y}(x) \rightarrow c_0(x))\\
\begin{split}
& \smash{\bigwedge_{(c, d) \not \in H}}
\forall x_1 (\vertex(x_1) \wedge d(x_1) \rightarrow\\ 
& \qquad \qquad \qquad \forall x_2 (\vertex(x_2) \wedge \kpred{X}(x_1, x_2)
\wedge \keq{Y}(x_1, x_2) \rightarrow  \neg c(x_2)))\\
& \smash{\bigwedge_{(c,d) \not \in V}}
\forall x_1 (\vertex(x_1) \wedge d(x_1) \rightarrow\\ 
& \qquad \qquad \qquad \forall x_2 (\vertex(x_2) \wedge \kpred{Y}(x_1, x_2)
\wedge \keq{X}(x_1, x_2) \rightarrow  \neg c(x_2)))
\end{split}
\end{align*}

This completes the construction of $\Psi_n$. We have shown that, if $\Psi_n$ is satisfiable, then
the $\tower{m}{n} \times \tower{m}{n}$-grid colouring problem $(C,c_0,H,V)$ has a solution. Conversely,
a simple check shows that, if the $\tower{m}{n} \times \tower{m}{n}$-grid colouring problem $(C,c_0,H,V)$ has a solution, then by interpreting the predicates involved in $\Psi_n$ as suggested above over a two-dimensional
toroidal grid of $m$-integers, we obtain a model of $\Psi_n$. This completes the reduction.
\end{proof}

\section{Upper bound}
\label{sec:upper}
In this section we first show that $\FL^3$ has the exponential model property, and later we prove a bounded model property for every $\FL^k$ with $k>3$. These results give the upper complexity bounds for the satisfiability problem for every $\FL^k$ with $k\geq 3$. 

It will be convenient to work with fluted formulas having a special form. Fix some purely relational signature $\sigma$ and some positive integer $k$.
A {\em fluted $k$-clause} is a disjunction of fluted $k$-literals. We allow the absurd formula $\bot$ (i.e.~the empty disjunction) to count as a fluted $k$-clause. Thus, any literal of a fluted $k$-clause is either 0-ary (i.e.~a proposition letter or its negation) or
has arguments $x_h, \dots, x_k$, in that order, for some $h$ ($1 \leq h \leq k)$.
When writing fluted $k$-clauses, we silently remove bracketing, re-order literals and delete duplicated literals as necessary.
It is easy to see that the number of fluted $k$-clauses, modulo these operations, is precisely $2^{2|\sigma|}$. 
To reduce notational clutter, we identify sets of clauses with their
conjunctions where convenient, writing $\Gamma$ instead of the more correct $\bigwedge \Gamma$.

A formula $\phi$ 
of $\FL^k$ ($k \geq 1$) is in \textit{normal form} if it is of the form
\begin{multline}
\forall x_1 \cdots x_{k} . \Omega \ \wedge\\
{\bigwedge_{i=1}^{s} }\forall x_1 \cdots x_{k-1} \left(\alpha_i \rightarrow \exists x_k . \Gamma_i\right) \wedge
{\bigwedge_{j=1}^{t}} \forall x_1 \cdots x_{k-1} (\beta_j \rightarrow \forall x_k . \delta_j),
\label{eq:nf}
\end{multline}
where $\Omega, \Gamma_1, \dots, \Gamma_s$ are sets of fluted $k$-clauses,
$\alpha_1, \dots, \alpha_s$, $\beta_1, \dots, \beta_t$ fluted \mbox{$(k-1)$}-atoms, 
and $\delta_1, \dots, \delta_t$ 
fluted $k$-clauses. We refer to $\forall x_1 \cdots x_{k} . \Omega$ as the \textit{static conjunct}
of $\phi$, to conjuncts of the form $\forall x_1 \cdots x_{k-1} \left(\alpha_i \rightarrow \exists x_k \Gamma_i\right)$ as the \textit{existential conjuncts}
of $\phi$, and to conjuncts of the form $\forall x_1 \cdots x_{k-1} (\beta_j \rightarrow \forall x_k . \delta_j)$ as the \textit{universal conjuncts} of $\phi$.

\begin{lemma}
	Let $\phi$ be an $\FL^{k}$-formula, $k \geq 1$. We can compute, in time bounded by a polynomial function of $\sizeOf{\phi}$, an $\FL^{k}$-formula $\psi$ in normal form such that:
	\textup{(}i\textup{)} $\models \psi \rightarrow \phi$; and \textup{(}ii\textup{)} any model of $\phi$ can be expanded to a model of $\psi$. 
	\label{lma:nf}
\end{lemma}
\begin{proof}
	By moving negations inward in the usual way, and applying standard re-writing techniques.
\end{proof}

We begin by showing that $\FL^3$ has the exponential-sized model property.
Here, the normal-form~\eqref{eq:nf} becomes 
\begin{equation}
\forall x_1 x_2 x_3 . \Omega \ \wedge 
\bigwedge_{i=1}^{s} \forall x_1 x_{2} \left(\alpha_i \rightarrow \exists x_3 . \Gamma_i\right) \wedge
\bigwedge_{j=1}^{t} \forall x_1 x_{2} (\beta_j \rightarrow \forall x_3 . \delta_j).
\label{eq:nf3}
\end{equation}
In the sequel, we shall additionally suppose that formulas do not contain any proposition letters. After all, when testing for satisfiability, the relevant truth-values can be simply guessed and eliminated accordingly. Since the
complexities involved are all super-exponential, such a guessing processes can be carried out without additional cost.

For $k \geq 1$, and relational signature $\sigma$, a {\em fluted $k$-type} \textit{over} $\sigma$ is a maximal, consistent set of fluted $k$-literals \textit{over} $\sigma$. As with sets of $k$-clauses, we identify fluted $k$-types with their conjunctions when convenient. If $\fA$ is a structure and $\bar{a}$ a tuple from $A$, then
there exists a unique fluted $k$-type $\tau$ satisfied by $\bar{a}$ in $\fA$: we call $\tau$ the 
fluted $k$-type of $\bar{a}$, and denote it by $\ftp^\fA[\bar{a}]$. If $\tau$ is a $k$-type, we denote by $\tau^{[1]}$ the result of incrementing the indices of all the variables in $\tau$, and if, in addition, $k \geq 2$, we denote by 
$\tau_\uparrow$ the result of removing all literals featuring the variable $x_1$ (i.e.~all literals of arity $k$) and decrementing the indices of all variables. Notice that $\tau_\uparrow$ will be a $(k-1)$-type; however,
$\tau^{[1]}$ will not be a $(k+1)$-type over $\sigma$ if $\sigma$ features any predicates of arity $(k+1)$.

A {\em connector-type} (\textit{over} $\sigma$) is a triple $\langle \pi, I,O \rangle$, where $\pi$ is a 1-type over $\sigma$,
$I$ is a set of fluted 2-types over $\sigma$ such that, for all $\tau \in I$, $\tau_\uparrow = \pi$, and
$O$ is a set of fluted 2-types over $\sigma$. We refer to $I$ as the connector-type's {\em inputs}, and to $O$ as its {\em outputs}. If $\fA$ is any structure interpreting $\sigma$, and $b \in A$, we define
\begin{equation*}
\ConA[b] = \langle \tpA[b], \{\ftp^\fA[a,b] \mid  a \in A\}, \{\ftp^\fA[b,c] \mid  c \in A\} \rangle.
\end{equation*}
It is obvious that $\ConA[b]$ is a connector-type; we call it {\em the connector-type of $b$ in $\fA$}. 

Let $\phi$ be a formula of the form~\eqref{eq:nf3} over some signature $\sigma$. A connector-type $\fc = \langle \pi, I,O \rangle$ over $\sigma$ is said to be \textit{locally compatible} \textit{with} $\phi$ if the following conditions hold:
\begin{description}
	\item[LC$\exists$] for every $i$ ($1 \leq i \leq  s$) and every $\tau \in I$ such that $\models \tau \rightarrow \alpha_i$, there exists $\tau' \in O$ such that the formula $\psi_i(x_1,x_2,x_3)$ is consistent, where
	\begin{equation*}
	\psi_i(x_1, x_2, x_3):= \tau \wedge {\tau'}^{[1]} \wedge \Gamma_i \wedge \Omega \wedge
	\bigwedge_{j=1}^{t} (\beta_j \rightarrow  \delta_j);
	\end{equation*}
	
	\item[LC$\forall$] for every $\tau \in I$ and every $\tau' \in O$, the formula $\psi(x_1,x_2,x_3)$ is consistent, where
	\begin{equation*}
	\psi(x_1, x_2, x_3):= \tau \wedge {\tau'}^{[1]} \wedge \Omega \wedge
	\bigwedge_{j=1}^{t} (\beta_j \rightarrow  \delta_j).
	\end{equation*}
\end{description}
\begin{lemma}
	If $\fA \models \phi$ and $b \in A$, then $\ConA[b]$ is locally compatible with $\phi$.
	\label{lma:compatibleEasy}
\end{lemma}
\begin{proof}
	Let $\ConA[b] = \langle \pi, I, O \rangle$.
	For LC$\exists$, pick $\tau \in I$ and $i$ ($1 \leq i \leq s$). From the construction of $I$,
	select $a \in A$ such that $\fA \models \tau[a,b]$. If $\models \tau \rightarrow \alpha_i$, then, since $\fA \models \phi$, select $c \in A$ such that the triple $a,b,c$ satisfies $\Gamma_i \wedge
	\Omega \wedge \bigwedge_{j=1}^{t} (\beta_j \rightarrow  \delta_j)$ in $\fA$, and set
	$\tau' = \ftp^\fA[b,c]$. Condition LC$\forall$ is treated similarly.
\end{proof}

A set $C$ of connector-types is said to be \textit{globally coherent} if the following conditions hold:
\begin{description}
	\item[GC$\exists$]
	for all $\fc = \langle \pi, I,O \rangle$ and all $\tau \in O$, there exists $\fd = \langle \pi', I',O' \rangle \in C$
	such that $\tau \in I'$;
	\item[GC$\forall$]
	for all $\fc = \langle \pi, I,O \rangle$ and all $\fd = \langle \pi', I',O' \rangle \in C$,
	$O \cap I' \neq \emptyset$.
\end{description}

\begin{lemma}
	If $\fA$ is a structure interpreting $\sigma$, then $\{\ConA[b] \mid b \in A\}$ is globally coherent.
	\label{lma:coherentEasy}
\end{lemma}
\begin{proof}
	Immediate from the definition of $\ConA[b]$.
\end{proof}

The relations between connector-types required by global coherence involve only (fluted) 2-types. This fact underlies the following lemma.
\begin{lemma}
	If $C$ is a globally coherent, non-empty set of connector-types, there exists a globally coherent, non-empty
	subset  $D \subseteq C$ 
	of cardinality at most $2^{|\sigma|}$.
	\label{lma:small}
\end{lemma}
\begin{proof}
	Pick any connector-type in $C$, initialize $D$ to be the singleton containing this
	connector-type, and initialize $I^*$ to be the set of its inputs, $I$. 
	We shall add connector-types to the set $D$ and fluted 2-types to the set $I^*$, maintaining the invariant that $I^*$ is the union of all the inputs of the connector-types in $D$.
	We call a connector-type in
	$D$ {\em satisfied} if its outputs are also included in the set $I^*$. 
	
	Now execute the following procedure until $D$ contains no unsatisfied \linebreak 
	connector-types.
	Pick some unsatisfied $\langle \pi, I,O \rangle \in D$ and some $\tau \in O \setminus I^*$. 
	By {\bf GC}$\exists$, there exists a connector-type $\langle \pi', I',O' \rangle \in C$ such that $\tau \in I$. Set $D := D \cup \{\langle \pi', I',O' \rangle\}$ and $I^* := I^* \cup I'$. These assignments maintain the invariant on $D$ and $I^*$. This process terminates after, say, $h \leq 2^{|\sigma|}$ steps, since there are only $2^{|\sigma|}$ fluted 2-types over $\sigma$, and $|I^*|$ increases by at least one at each step.
	When it does so, $D$ is globally coherent, and indeed $|D| \leq h$.
\end{proof}

\begin{lemma}
	Let $\phi \in FL^3$ be in normal form over a signature $\sigma$, and let $\phi$ have $s$ existential conjuncts. The following are equivalent:
	\textup{(}i\textup{)}
	$\phi$ is satisfiable;
	\textup{(}ii\textup{)}
	there exists a non-empty, globally coherent set $C$ of connector-types locally compatible with $\phi$
	such that $|C| \leq 2^{|\sigma|}$;
	\textup{(}iii\textup{)}
	$\phi$ is satisfiable over a domain of size at most $s \cdot 2^{2|\sigma|}$.
	\label{lma:construction}
\end{lemma}
\begin{proof}
	The implication (i) $\Rightarrow$ (ii) follows from
	Lemmas~\ref{lma:compatibleEasy}, \ref{lma:coherentEasy} and~\ref{lma:small}. The implication (iii) $\Rightarrow$ (i) is trivial. It remains only to show (ii) $\Rightarrow$ (iii). 
	
	Suppose that $C$ is a non-empty, globally coherent set of connector-types locally compatible with $\phi$
	and that $|C| \leq 2^{|\sigma|}$. Denote the set of fluted atomic 2-types over $\sigma$ by
	$\ftp^2(\sigma)$. For all $\fc \in C$
	and all $\tau \in \ftp^2(\sigma)$, let 
	$A_{\fc,\tau}$ be a fresh set of $s$ elements, $A_{\fc,\tau} = \{a_{\fc,\tau,1}, \dots, a_{\fc,\tau, s}\}$. For each $\fc \in C$, let $A_\fc = \bigcup_{\tau \in \ftp^2(\sigma)} A_{\fc,\tau}$; and let $A = \bigcup_{\fc \in C} A_{\fc}$.
	Hence $|A| \leq  s \cdot 2^{|\sigma|}\cdot |C| \leq  s \cdot 2^{2|\sigma|}$.
	
	For every $\fc= \langle \pi, I,O \rangle  \in C$ and every $\tau \in O$, 
	pick some $\fd= \langle \pi', I',O' \rangle \in C$ such that $\tau \in I' $, by {\bf GC}$\exists$, and then,
	for every $a \in A_\fc$ and every $i$ ($1 \leq i \leq s$), set $\tpA[a,a_{\fd,\tau,i}] = \tau$. The conditions on connector-types ensure that this assignment
	of fluted 2-types is consistent with the 1-type $\pi'$ assigned to $a_{\fd,\tau,i}$; moreover, no clashes can arise by 
	the disjointness of the sets $A_{\fd,\tau}$.
	These assignments having been made, we complete the assignment of 2-types in $\fA$ as follows.
	As long as there exist $a \in A_\fc$ and
	$b \in A_\fd$, with $\fc= \langle \pi, I,O \rangle  \in C$ and $\fd= \langle \pi', I',O' \rangle \in C$, such that
	$\tpA[a,b]$ is not yet defined, pick some $\tau \in O \cap I'$, by {\bf GC}$\forall$, and
	set $\tpA[a,b] = \tau$. Again, the conditions on connector-types ensure that this assignment
	is consistent with the 1-type $\pi'$ assigned to $b$.
	
	At the end of this process, all 2-types have been defined in such a way that, if
	$b \in A_\fc$, where $\fc= \langle \pi, I,O \rangle  \in C$, then:
	\begin{description}
		\item[P1] for all $a \in A$, $\tpA[a,b] \in I$;
		\item[P2] for all $\tau \in O$, there exists a $\fd \in C$ such that, for each of the $s$ elements $c \in A_{\fd,\tau}$, $\tpA[b,c] = \tau$;
		\item[P3] for all $c \in A$, 
		$\tpA[b,c] \in O$.
	\end{description}
	We now assign 3-types in $\fA$
	so as to ensure that $\fA \models \phi$. We deal first with the existential conjuncts. Let $a, b$ be any
	elements in $A$ (not necessarily distinct). Let $b \in A_{\fc}$, where $\fc = \langle \pi,I,O \rangle$, and let
	$\tpA[a,b] = \tau$. By {\bf P1}, $\tau \in I$, so that, by {\bf LC}$\exists$, for all $i$ ($1 \leq i \leq s$), 
	if $\models \tau \rightarrow \alpha_i$, then there exists $\tau' \in O$ such that 
	\begin{equation}
	\tau \wedge {\tau'}^{[1]} \wedge \Gamma_i \wedge \Omega \wedge
	\bigwedge_{j=1}^{t} (\beta_j \rightarrow  \delta_j)
	\label{eq:existsLocal}
	\end{equation}
	is consistent. By {\bf P2}, 
	pick some $\fd \in C$ such that, for each of the $s$ elements $c \in A_{\fd,\tau'}$, $\tpA[b,c] = \tau'$.
	Choose some fluted 3-type $\xi$ consistent with~\eqref{eq:existsLocal}, and set $\tpA[a,b,a_{\fd,i}] = \xi$.
	Clearly, we may carry out this process for each $i$ ($1 \leq i \leq s$), and indeed for each  pair of
	elements $a, b \in A$, without clashes. At the end of this process, all existential conjuncts of $\phi$ will
	be satisfied (irrespective of how the model is completed): moreover, none of the fluted 3-types that have been set violates any of the static or universal conjuncts of $\phi$.
	
	Finally, suppose $a$, $b$ and $c$ are elements of $A$ (not necessarily distinct) such that $\tpA[a,b,c]$ has not been
	defined. Let $\tpA[a,b]= \tau$ and $\tpA[b,c]= \tau'$. Suppose $b \in A_\fc$, where 
	$\fc= \langle \pi, I,O \rangle$, so that, by {\bf P1} and {\bf P3}, $\tau \in I$ and $\tau' \in O$. By
	{\bf LC}$\forall$,
	\begin{equation}
	\tau \wedge {\tau'}^{[1]} \wedge \Omega \wedge
	\bigwedge_{j=1}^{t} (\beta_j \rightarrow  \delta_j)
	\label{eq:forallLocal}
	\end{equation}
	is consistent. Choose some fluted 3-type $\xi$ consistent with~\eqref{eq:forallLocal}, and set\linebreak $\tpA[a,b,c] = \xi$. At the end of this process, $\fA$ will
	be fully defined: moreover, none of the fluted 3-types that have been set violates any of the 
	static or universal conjuncts of $\phi$.
	Thus, $\fA \models \phi$ as required.
\end{proof}

We now turn our attention to obtaining a small model property for 
$\FL^m$ for all $m \geq 3$. Our strategy is to show that, given an $\FL^m$ formula $\phi$
with $s$ existential conjuncts, we
can compute an $\FL^{m-1}$ formula $\psi$ such that: if $\phi$ is satisfiable, then $\psi$
is satisfiable; and if $\psi$ has a model of size $M$, then $\psi$ has a model of size $sM$.
However, the size of $\psi$ and indeed the size of its signature, both increase by an exponential. 
We employ the standard apparatus of resolution theorem-proving. Fixing some relational signature, let $p$ be a predicate of arity $k$ and let $\gamma'$ and $\delta'$ be fluted $k$-clauses. Then, $\gamma = p(x_1, \dots, x_k) \vee \gamma'$ and $\delta = \neg p(x_1, \dots, x_k) \vee \delta'$ are also fluted $k$-clauses, as indeed is $\gamma' \vee \delta'$. In that case, we call $\gamma' \vee \delta'$ a {\em fluted resolvent} of $\gamma$ and $\delta$, and we say that $\gamma' \vee \delta'$ is {\em obtained by fluted resolution} from $\gamma$ and $\delta$ \textit{on} $p(x_1, \dots, x_k)$. Thus, fluted resolution 
is simply a restriction of the familiar resolution rule from first-order logic to the case where the resolved-on literals have maximal arity, $k$. It may be helpful to note the following at this point: 
(i) if $\gamma$ and $\delta$ resolve to form $\epsilon$, then $\models \gamma \wedge \delta \rightarrow \epsilon$;
(ii) the fluted resolvent of two clauses may or may not
involve predicates of arity $k$;
(iii) in fluted resolution, the arguments of the literals in the clauses undergo no change when forming the resolvent;
(iv) if the fluted $k$-clause $\gamma$ involves no predicates
of arity $k$, then it cannot undergo fluted resolution at all.

If $\Gamma$ is a set of fluted $k$-clauses, denote by $\Gamma^*$ the smallest set of fluted $k$-clauses including $\Gamma$ and closed under fluted resolution. If $\Gamma = \Gamma^*$, we say that it is \textit{closed under fluted resolution}. We further denote by $\Gamma^\circ$ the result of deleting
from $\Gamma^*$ any clause involving a predicate of arity $k$. Observe that $\Gamma^\circ$ does not feature the variable $x_1$.

The following lemma is, in effect, nothing more than the familiar completeness theorem for (ordered) propositional resolution.
\begin{lemma}
	Let $\Gamma$ be a set of fluted $k$-clauses, and $\tau$ a fluted $(k-1)$-type.  If $\Gamma^\circ 
	\cup \{\tau(x_2, \dots, x_k)\}$ is consistent, then there exists a fluted $k$-type $\tau^+$ such that
	$\tau^+ \supseteq \tau(x_2, \dots, x_k)$ and $\Gamma \cup \{\tau^+\}$ is consistent.
	\label{lma:resolution}
\end{lemma}
\begin{proof}
	Enumerate the $k$-ary predicates occurring in $\Gamma$ as $p_1, \dots, p_n$. Note that none of these predicates
	occurs in $\tau$.
	Define a \textit{level-$i$ extension}
	of $\tau$ inductively as follows: (i) $\tau(x_2, \dots, x_k)$ is an level-0 extension of $\tau$; (ii) if $\tau'$ is a level-$i$ extension of $\tau$ ($0 \leq i < n$), then $\tau' \cup \{p_{i+1}(x_1, \dots, x_k)\}$  and $\tau' \cup \{\neg p_{i+1}(x_1, \dots, x_k)\}$ are 
	level-$(i+1)$ extensions of $\tau$. Thus, the level-$n$ extensions of $\tau$ are exactly 
	the fluted $k$-types over $\sigma$ extending
	$\tau(x_2, \dots, x_k)$. 
	If $\tau'$ is a level-$i$ extension of $\tau$ ($0 \leq i < n$), we say that $\tau'$ {\em violates} a clause $\delta$
	if, for every literal in $\gamma$, the opposite literal is in $\tau'$; we say that $\tau'$ {\em violates} a set of clauses
	$\Delta$ if $\tau'$ violates some $\delta \in \Delta$.
	Suppose now that $\tau'$ is a level-$i$ extension of $\tau$ ($0 \leq i < n$). We claim that, if both 
	$\tau' \cup \{p_{i+1}(x_1, \dots, x_k)\}$  and $\tau' \cup \{\neg p_{i+1}(x_1, \dots, x_k)\}$ violate $\Gamma^*$, then so does $\tau$. For otherwise, there must be a clause $\neg p_{i+1} \vee \gamma' \in \Gamma^*$ violated by $\tau' \cup \{p_{i+1}(x_1, \dots, x_k)\}$ and a clause $p_{i+1} \vee \gamma' \in \Gamma^*$ violated by $\tau' \cup \{\neg p_{i+1}(x_1, \dots, x_k)\}$. But in that case $\tau'$ violates the fluted resolvent $\gamma' \vee \delta'$,  contradicting the
	supposition that $\tau'$ does not violate $\Gamma^*$. This proves the claim. Now, since $\tau(x_2, \dots, x_{k})$ 
	is by hypothesis consistent with $\Gamma^\circ$, it certainly does not violate $\Gamma^\circ$. Moreover, since it involves no predicates of
	arity $k$, $\tau(x_2, \dots, x_{k})$ does not violate $\Gamma^*$ either. By the above claim, then,
	there must be at least one level-$n$ extension of $\tau$ which does not
	violate $\Gamma^* \supseteq \Gamma$. Since $\tau^+$ is a $k$-type, this proves the lemma.
\end{proof}

Before giving the decision procedure for $\FL^{k}$, we require one simple technical lemma.
\begin{lemma}
	Let $\fA$ be any structure, and let $z >0$. There exists a structure $\fB$ such that \textup{(}i\textup{)} if $\phi$ is any first-order formula without equality, then
	$\fA \models \phi$ if and only if $\fB \models \phi$;  \textup{(}ii\textup{)} 
	$|B| = z \cdot |A|$; and
	\textup{(}iii\textup{)}
	if $\psi(x_1, \dots, x_{k-1}) = \exists x_k \chi(x_1, \dots, x_k)$ is a first-order formula without equality, and $\fB \models \psi[b_1, \dots, b_{k-1}]$, then there exist at least $z$ distinct elements $b$ of $B$ such that $\fB \models \chi[b_1, \dots, b_{k-1}, b]$. 
	\label{lma:multiply}
\end{lemma}
\begin{proof}
	Let $B = \{1, \dots, z\} \times A$. If $p$ is any predicate of arity $k$, set\linebreak
	$\langle (i_i, a_1), \dots, (i_k, a_k) \rangle \in p^\fB$ if and only if
	$\langle a_1, \dots, a_k \rangle \in p^\fA$. It is simple to verify that $\fB$ has the desired properties. 
\end{proof}

\begin{lemma}
	Let  $\phi$ be a normal-form $\FL^{m}$-formula \textup{(}$m \geq 3$\textup{)} 
	over some relational signature $\sigma$, and suppose that $\phi$ has $s$ existential conjuncts
	and $t$ universal conjuncts. 
	Then there exists a normal-form $\FL^{m-1}$-formula $\phi'$ over a relational signature $\sigma'$ such that the following hold:
	\textup{(}i\textup{)} $\phi'$ has at most $2^{t} s$ existential and $2^{t}$ universal conjuncts;
	\textup{(}ii\textup{)} $|\sigma'| \leq |\sigma| + 2^t(s+1)$;
	\textup{(}iii\textup{)} if $\phi$ has a model, so does $\phi'$; and
	\textup{(}iv\textup{)} if $\phi'$ has a model of size $M$, then $\phi$ has a model of size at most $sM$.
	\label{lma:eliminate}
\end{lemma}
\begin{proof}
	We repeat the form of $\phi$ given in~\eqref{eq:nf} for convenience:
	\begin{multline*}
	\forall x_1 \cdots x_{k} . \Omega \ \wedge
	{\bigwedge_{i=1}^{s} }\forall x_1 \cdots x_{k-1} \left(\alpha_i \rightarrow \exists x_k . \Gamma_i\right) \wedge
	{\bigwedge_{j=1}^{t}} \forall x_1 \cdots x_{k-1} (\beta_j \rightarrow \forall x_k . \delta_j).
	\end{multline*}
	Write $T = \{1, \dots, t\}$. For
	all $i$ ($1 \leq i \leq s$) and all $J \subseteq T$, let $p_{i,J}$ and $q_J$ be new
	predicates of arity $m-2$. Let $\phi'$ be the conjunction of the sentences
	\begin{align}
	&
	\bigwedge_{i=1}^{s} \bigwedge_{J \subseteq T} \forall x_1 \cdots x_{m-1}
	((\alpha_i \wedge \bigwedge_{j \in J} \beta_j )  \rightarrow p_{i,J}(x_2, \dots, x_{m-1}))
	\label{eq:phiStar1}\\
	&
	\bigwedge_{J \subseteq T}\forall x_1 \cdots x_{m-1}
	((\bigwedge_{j \in J} \beta_j) \rightarrow q_{J}(x_2, \dots, x_{m-1}))
	\label{eq:phiStar2}\\
	\begin{split}
	& 
	\bigwedge_{i=1}^{s} \bigwedge_{J \subseteq T} \forall x_2 \cdots x_{m-1}
	\left( p_{i,J}(x_2, \dots, x_{m-1})\rightarrow \right.\\
	& \hspace{3.25cm} \left. \exists x_m \left(\Gamma_i \cup \Omega \cup \{\delta_j \mid j \in J \}\right)^\circ \right) 
	\end{split} 
	\label{eq:phiStar3}\\ 
	&
	\bigwedge_{J \subseteq T}\forall x_2 \cdots x_{m-1}
	\left( q_{J}(x_2, \dots, x_{m-1})\rightarrow \forall x_m \left(\Omega \cup  \{\delta_j \mid j \in J \}\right)^\circ \right).
	\label{eq:phiStar4}
	\end{align}
	We claim that, if $\phi$ is satisfiable, then so is $\phi'$. 
	For suppose $\fA \models \phi$. We expand
	$\fA$ to a model $\fA' \models \phi'$ by setting, for all $i$ ($1 \leq i \leq s$) and all $J \subseteq T$,
	\begin{align*}
	\begin{split}
	p_{i,J}^{\fA'}  = \{ \langle a_2, & \dots, a_{m-1} \rangle \mid 
	\text{ for some $a_1 \in A$: }\\ 
	& \text{$\fA \models \alpha_i[a_1, \dots, a_{m-1}]$ and 
		$\fA \models \beta_j[a_1, \dots, a_{m-1}]$ for all $j \in J$}  \}
	\end{split}\\
	\begin{split}
	q_{J}^{\fA'}  = \{ \langle a_2, & \dots, a_{m-1} \rangle \mid 
	\text{ for some $a_1 \in A$: }\\ 
	& \text{$\fA \models \beta_j[a_1, \dots, a_{m-1}]$ for all $j \in J$}  \}.
	\end{split}
	\end{align*}
	To see that $\fA' \models \phi'$, we simply check the truth of conjuncts~\eqref{eq:phiStar1}--\eqref{eq:phiStar4} in $\fA'$ in turn.
	Sentences~\eqref{eq:phiStar1} and~\eqref{eq:phiStar2} are immediate. For~\eqref{eq:phiStar3}, fix $i$ and $J$, and suppose
	$\fA' \models p_{i,J}[a_2, \dots, a_{m-1}]$. By the definition of $\fA'$,
	let $a_1 \in A$ be such that $\fA \models \alpha_i[a_1, \dots, a_{m-1}]$ and 
	$\fA \models \beta_j[a_1, \dots, a_{m-1}]$ for all $j \in J$. Since $\fA \models \phi$, there exists
	$b$ such that $\fA \models \Gamma_i[a_1, \dots, a_{m-1},b]$, \mbox{$\fA \models \Omega[a_1, \dots, a_{m-1},b]$} and\linebreak
	$\fA \models \delta_j[a_1, \dots, a_{m-1},b]$ for all $j \in J$. Since resolution is a valid inference step,
	$\fA' \models \left(\Gamma_i \cup \Omega \cup \{\delta_j \mid j \in J \}\right)^\circ[a_2, \dots, a_{m-1},b]$. This
	establishes the truth of~\eqref{eq:phiStar3} in $\fA'$. Sentence~\eqref{eq:phiStar4} is handled similarly.
	
	Conversely, we claim that, if 
	$\phi'$ is satisfiable over a domain $A$, then $\phi$  is satisfiable over a domain of size  $s \cdot |A|$.
	For suppose $\fA \models \phi'$. Let $\fB$ be the model of $\phi'$ guaranteed by Lemma~\ref{lma:multiply}, where $z= s$. We may assume that $\fA$ and hence $\fB$ interpret no predicates of arity $m$.
	We proceed to expand $\fB$ to a model $\fB' \models \phi$ by interpreting the predicates of arity $m$ occurring in $\phi$.      
	Pick any tuple $\langle a_1, a_2, \dots, a_{m-1} \rangle$ from $B$, and let $J$ be the set of all $j$ ($1 \leq j \leq t$)
	such that $\fB \models \beta_j[a_1, a_2, \dots, a_{m-1}]$. Suppose also that, for some $i$ ($1 \leq i \leq s$)
	$\fA \models \alpha_i[a_1, a_2, \dots, a_{m-1}]$. 
	From~\eqref{eq:phiStar1}, $\fB \models p_{i,J}[a_2, \dots, a_{m-1}]$; and
	from~\eqref{eq:phiStar3},
	we may pick $b_i \in B$ such that
	$\fB \models \left(\Gamma_i \cup \Omega \cup \{\delta_j \mid j \in J \}\right)^\circ[a_2, \dots, a_{m-1},b_i]$. From the properties secured for $\fB$ by Lemma~\ref{lma:multiply}, we know that if, for fixed $a_1, \dots, a_{m-1}$,
	we have $\fA \models \alpha_i[a_1, a_2, \dots, a_{m-1}]$ for more than one value of $i$, then
	we may choose the corresponding elements $b_{i}$ so that they are all distinct.  
	For each such $b_i$, then, let $\tau_i = \tpB[a_2, \dots, a_{m-1},b_{i}]$. Thus, $\tau_i(x_2, \dots, x_m)$ is consistent
	with $\left(\Gamma_i \cup \Omega \cup \{\delta_j \mid j \in J \}\right)^\circ$. By Lemma~\ref{lma:resolution}, there exists a 
	fluted $m$-type $\tau_i^+ \supseteq \tau_i(x_2, \dots, x_m)$ such that $\tau_i^+$ is consistent with $\Gamma_i \cup \Omega \cup  \{\delta_j \mid j \in J \}$. Set $\tpBp[a_1,a_2, \dots, a_{m-1},b_{i}] = \tau^+_i$. Since $\tau^+_i \supseteq \tau_i(x_2, \dots, x_m)$,
	only predicates of arity $m$ are being assigned, so that there is no clash with $\fB$. Moreover, since the $b_i$ are all distinct,
	for a given tuple $a_1, \dots, a_{m-1}$, these assignments do not clash with each other. In this way, every existential conjunct of $\phi$ is witnessed in $\fB'$ for the tuple
	$a_1,a_2, \dots, a_{m-1}$, and no static or universal conjunct of $\phi$ is violated for the tuples from $B$ for which the $m$-ary predicates of $\sigma$ have been defined.
	Now let $\langle a_1, a_2, \dots, a_{m-1}, a_m \rangle$ be any tuple from $B$ for which the $m$-ary predicates
	of $\sigma$ have not been defined, and let $J$ be the set of all $j$ ($1 \leq j \leq m$)
	such that $\fB \models \beta_j[a_1, a_2, \dots, a_{m-1}]$. From~\eqref{eq:phiStar2},
	$\fB \models q_J[a_2, \dots, a_{m-1}]$; and from~\eqref{eq:phiStar4},\linebreak $\fB \models \left(\Omega \cup  \{\delta_j \mid j \in J \} \right)^\circ[a_2, \dots, a_{m-1}, a_m]$. 
	Let $\tau= \tpB[a_2, \dots, a_{m-1},a_m]$. 
	Hence $\tau(x_2, \dots, x_m)$ is consistent with 
	$\left( \{\delta_j \mid j \in J \} \right)^\circ$. By Lemma~\ref{lma:resolution}, 
	there exists a 
	fluted $m$-type $\tau^+ \supseteq \tau(x_2, \dots, x_m)$ such that $\tau^+$ is consistent 
	with $\Omega \cup  \{\delta_j \mid j \in J \}$. Set $\tpBp[a_1,a_2, \dots, a_{m-1},a_m] = \tau^+$. 
	Since $\tau^+ \supseteq \tau(x_2, \dots, x_m)$,
	only predicates of arity $m$ are being assigned, so that there is no clash with $\fB$. Evidently, no static or universal conjunct of $\phi$ is violated in this process. Thus, $\fB' \models \phi$, as required.
	
	Conjuncts~\eqref{eq:phiStar1} and~\eqref{eq:phiStar2} do not involve $x_m$, while~\eqref{eq:phiStar3} and~\eqref{eq:phiStar4} do not involve $x_1$. Thus, by decrementing variable indices where necessary,
	$\phi'$ becomes a formula of $\FL^{m-1}$. Indeed, re-writing implications as clauses in the
	obvious way, we have a formula in normal form: formulas
	\eqref{eq:phiStar1} and~\eqref{eq:phiStar2} give the static conjunct; formula~\eqref{eq:phiStar3} gives the $2^t s$
	existential conjuncts; and formula~\eqref{eq:phiStar4} gives the $2^t$ 
	universal conjuncts. Exactly $2^t(s +1)$ predicates of arity \mbox{${(m-2)}$} have been added to the signature, and some $m$-ary predicates may have been lost.
\end{proof}

\begin{lemma}
	Let $\phi$ be a normal-form formula of $\FL^m$ \textup{(}$m\geq 3$\textup{)} over a signature $\sigma$, and suppose that $\phi$ has
	$s$ existential conjuncts \textup{(}$s \geq 3$\textup{)} and $t$ universal conjuncts. If $\phi$ is satisfiable, then it is satisfiable over a domain of size at most $\tower{m-2}{2|\sigma|+s+t}$. 
	\label{lma:smallModels}
\end{lemma}
\begin{proof}
	We proceed by induction on $m$. If $m=3$, Lemma~\ref{lma:construction} guarantees that 
	$\phi$ has a model of size $s \cdot 2^{2|\sigma|} \leq 2^{2|\sigma|+s+t}$.
	
	Suppose, then $m>3$, and that the lemma holds for smaller values of $m$. Let $\phi'$ be the
	formula guaranteed by Lemma~\ref{lma:eliminate}, so that $\phi'$ is satisfiable.
	By inductive hypothesis $\phi'$ has a model of size at most $M=\tower{m-3}{2(|\sigma|+2^t(s+1))+2^t s+2^t}$,
	whence, by property (iv) of Lemma~\ref{lma:eliminate},
	$\phi$ has a model of size at most 
	\begin{align*}
	s\cdot\tower{m-3}{2(|\sigma|+2^t(1+s))+2^t s+2^t} & = s\cdot\tower{m-3}{2|\sigma|+3\cdot 2^t(1+s)}\\
	& \leq \tower{m-3}{2|\sigma|+3\cdot 2^t(1+s)+s}.
	\end{align*}
	A routine calculation shows that, for $s \geq 3$, $|\sigma| \geq 1$ and $t \geq 0$,
	$2|\sigma|+3\cdot 2^t(1+s)+s \leq 2^{2|\sigma| + s+ t}$. 
	This completes the inductive step, and proves the lemma.
\end{proof}
% Irritating calculations for the inequality
%
% Write z = 2^t and x = 2|\sigma|. Thus x >= 2,
% and >= 1. Assume also s >= 3. 
%
% When t = 0, s = 3, LHS  is x + 15 and RHS is 8. 2^x , 
% so LHS <= RHS. Also, when t = 0 LHS is x + 4s + 3 and
% RHS is 2^{x+s} so incrementing s increases LSH by 4 and 
% RHS by 2^x >= 4, so, if t = 0, the inequality holds for all
% x and s >=3. 
%
% Now incrementing t doubles z and so increases LHS by 3z(1+s)
% and RHS by 2^{x+s}z. So it suffices to show that 3.(1+s) <= 
% 2^{x+s} for all x >= 2 and s >= 3. This is easily checked for 
% s = 3. But incrementing s increases LHS by 3 and doubles RHS,
% so this is true. 

Lemmas~\ref{lma:nf} and~\ref{lma:smallModels} thus yield the following upper complexity bounds.
\begin{theorem}
	If $m \geq 3$, then any satisfiable formula of $\FL^{m}$ has a model of
	$(m-2)$-tuply exponential size. Hence,
	the satisfiability problem for $\FL^{m}$ is in $(m-2)$-\NExpTime.
	\label{theo:upper}
\end{theorem}

Thus, for $m \geq 3$, the satisfiability problem for $\FL^m$ lies between $\lfloor m/2 \rfloor$-
\linebreak
\NExpTime-hard
and $(m-2)$-\NExpTime. It is conceivable that, by using appropriate data-structures in place of connector-types, Lemma~\ref{lma:eliminate}
might be generalized to yield an improved upper-bound for all values of $m$. The present authors have, however, been unable to do so, even for the value $m=5$. Small-model properties---and hence upper complexity bounds---for values of $m$ up to 2 are easily derivable from known results. Trivially,
all satisfiable $\FL^0$-formulas have models with 1-element domains; and 
since $\FL^1$
is identical to the 1-variable fragment of first-order logic, any satisfiable $\FL^1$-formula $\phi$ has a model of size bounded by $\sizeOf{\phi}$. Moreover, $\FL^2$ is contained within the 2-variable fragment of first-order logic, which is shown to have the exponential-sized model property by Gr\"{a}del, Kolaitis and Vardi~\cite{purdy:gkv97}. (Note that,
for the fluted fragment, Lemma~\ref{lma:construction} above strengthens this to three variables.) 
Taking ``0-\NExpTime'' to mean ``\NPTime'', and noting that $4 -2 = 4/2$, we see that $\FL^m$ is 
$\lfloor m/2 \rfloor$-\NExpTime-complete for all $m$ up to the value 4.

\bibliographystyle{plain}

\begin{thebibliography}{10}
	
	\bibitem{purdy:avbn98}
	H.~Andr{\'{e}}ka, J.~van Benthem, and I.~N{\'{e}}meti.
	\newblock Modal languages and bounded fragments of predicate logic.
	\newblock {\em Journal of Philosophical Logic}, 27(3):217--274, 1998.
	
	\bibitem{purdy:BGG97}
	E.~B{\"o}rger, E.~Gr{\"a}del, and Y.~Gurevich.
	\newblock {\em The Classical Decision Problem}.
	\newblock Springer, 1997.
	
	\bibitem{purdy:gkv97}
	E.~Gr{\"a}del, P.~Kolaitis, and M.~Vardi.
	\newblock On the decision problem for two-variable first-order logic.
	\newblock {\em Bulletin of Symbolic Logic}, 3(1):53--69, 1997.
	
	\bibitem{purdy:hsg04}
	U.~Hustadt, R.~Schmidt, and L.~Georgieva.
	\newblock A survey of decidable first-order fragments and description logics.
	\newblock {\em Journal of Relational Methods in Computer Science},
	1(3):251--276, 2004.
	
	\bibitem{purdy:ls01}
	C.~Lutz and U.~Sattler.
	\newblock The complexity of reasoning with {B}oolean modal logics.
	\newblock In F.~Wolter, H.~Wansing, M.~de~Rijke, and M.~Zakharyaschev, editors,
	{\em Advances in Modal Logic}, volume~3, pages 1--20, Menlo Park, 2001. CLSI
	Publications.
	
	\bibitem{purdy:noah80}
	A.~Noah.
	\newblock Predicate-functors and the limits of decidability in logic.
	\newblock {\em Notre Dame Journal of Formal Logic}, 21(4):701--707, 1980.
	
	\bibitem{P-HST16}
	Ian Pratt{-}Hartmann, Wieslaw Szwast, and Lidia Tendera.
	\newblock Quine's fluted fragment is non-elementary.
	\newblock In {\em 25th {EACSL} Annual Conference on Computer Science Logic,
		{CSL} 2016, August 29 - September 1, 2016, Marseille, France}, pages
	39:1--39:21, 2016.
	
	\bibitem{purdy:purdy96b}
	W.~C. Purdy.
	\newblock Decidability of fluted logic with indentity.
	\newblock {\em Notre Dame Journal of Formal Logic}, 37(1):84--104, 1996.
	
	\bibitem{purdy:purdy96a}
	W.~C. Purdy.
	\newblock Fluted formulas and the limits of decidability.
	\newblock {\em Journal of Symbolic Logic}, 61(2):608--620, 1996.
	
	\bibitem{purdy:purdy99}
	W.~C. Purdy.
	\newblock Quine's limits of decision.
	\newblock {\em Journal of Symbolic Logic}, 64(4):1439--1466, 1999.
	
	\bibitem{purdy:purdy02}
	W.~C. Purdy.
	\newblock Complexity and nicety of fluted logic.
	\newblock {\em Studia Logica}, 71:177--198, 2002.
	
	\bibitem{purdy:quine69}
	W.~V. Quine.
	\newblock On the limits of decision.
	\newblock In {\em Proceedings of the 14th International Congress of
		Philosophy}, volume III, pages 57--62. University of Vienna, 1969.
	
	\bibitem{purdy:quine76a}
	W.~V. Quine.
	\newblock Algebraic logic and predicate functors.
	\newblock In {\em The Ways of Paradox}, pages 283--307. Harvard University
	Press, revised and enlarged edition, 1976.
	
	\bibitem{purdy:quine76b}
	W.~V. Quine.
	\newblock The variable.
	\newblock In {\em The Ways of Paradox}, pages 272--282. Harvard University
	Press, revised and enlarged edition, 1976.
	
	\bibitem{purdy:SH00}
	R.~Schmidt and U.~Hustadt.
	\newblock A resolution decision procedure for fluted logic.
	\newblock In D.~McAllester, editor, {\em Proceedings, CADE}, number 1831 in
	LNAI, pages 433--448. Springer-Verlag, 2000.
	
\end{thebibliography}

\end{document}